\documentclass[10pt,a4paper]{article}

\usepackage[a4paper,margin=1.5cm]{geometry}

\usepackage{amsmath,amssymb,amsthm}
\usepackage{graphicx}
\usepackage{booktabs}
\usepackage{array}
\usepackage{longtable}
\usepackage{multirow}
\usepackage{caption}
\usepackage{setspace}
\usepackage{enumitem}
\usepackage{hyperref}
\usepackage{xcolor}
\usepackage{pgfplots}
\usepackage{tikz}
\usepackage{pgfplotstable}
\usepackage{tikz}
\usepackage{amsthm}
\usepackage{array}
\usepackage{placeins}
\usepackage{cite}
\usepackage{float}

\theoremstyle{plain}
\newtheorem{theorem}{Theorem}[section]
\newtheorem{lemma}[theorem]{Lemma}
\newtheorem{corollary}[theorem]{Corollary}
\newtheorem{proposition}[theorem]{Proposition}

\theoremstyle{definition}

\theoremstyle{remark}

\usetikzlibrary{fillbetween}
\pgfplotsset{compat=1.18}

\usetikzlibrary{
arrows.meta,
calc,
positioning,
patterns,
fillbetween
}

\pgfplotsset{compat=1.18}

\hypersetup{
    colorlinks=true,
    linkcolor=blue,
    citecolor=blue,
    urlcolor=blue
}

\title{
\bfseries
Sensitivity of evolutionary entropy and reproductive potential in Lefkovitch and open-group Leslie matrices
}

\author{Henrique M. Oliveira$^{1,2}$}
\date{\today}

\begin{document}

\maketitle
\begin{center}
\small
$^{1}$ Departamento de Matemática, Instituto Superior Técnico,
Universidade de Lisboa, 1049-001 Lisboa, Portugal\\
$^{2}$ Center for Mathematical Analysis, Geometry and Dynamical Systems
(CAMGSD), Instituto Superior Técnico, Universidade de Lisboa,
1049-001 Lisboa, Portugal\\
E-mail: henrique.m.oliveira@tecnico.ulisboa.pt
\end{center}
\begin{abstract}
\begingroup
\color{black}

We develop a sensitivity theory for the Malthusian parameter
\(r\), evolutionary entropy \(H\), and reproductive potential
\(\Phi\) in Lefkovitch matrices, with open-group Leslie matrices as a
special case. Starting from the dominant eigenvalue and the
Markov chain associated with the projection matrix, we
derive explicit formulas for the long-run stage distribution,
generation time, evolutionary entropy, and reproductive potential.
Evolutionary entropy and reproductive potential both admit exact
transition--retention decompositions, with distinct interpretations:
for \(H\) the separation concerns sources of demographic uncertainty,
whereas for \(\Phi\) it partitions the demographic links themselves.

For small changes to positive coefficients, we obtain
closed-form sensitivities for fertility, progression, and retention.
The sensitivity of \(r\) follows from the standard formula for a
simple dominant eigenvalue, while that of \(\Phi\) follows from the
exact identity
\[
r = H + \Phi.
\]
The formulas specialise to open-group Leslie
matrices, in which retention occurs only in the terminal stage, and
reduce to the classical Leslie theory when retention vanishes.

We also distinguish a change to an existing demographic link from the
creation of a fertility or retention link whose coefficient was zero.
When a new link is created, \(r\) still has a finite one-sided
derivative, whereas the leading changes in \(H\) and \(\Phi\) are
opposite terms of order \(\varepsilon\log\varepsilon\). Empirical
examples show that growth, evolutionary entropy, and reproductive
potential can respond differently to the same demographic change.
They also show that stage retention may account for most of \(H\)
without determining whether \(H\) rises or falls when retention is
increased.

\endgroup
\end{abstract}


\medskip


\section{Introduction}\label{sec:1}

\begingroup
\color{black}
Structured population models describe populations whose demographic
rates depend on age, size, developmental stage or physiological state.
Since the pioneering work of Leslie
\cite{Leslie1945,Leslie1948}, population projection matrices have
become a fundamental tool for the analysis of population growth,
stable demographic structure and reproductive value. Lefkovitch
\cite{Lefkovitch1965} extended the classical age-structured model by
allowing individuals to remain within the same demographic stage for
more than one projection interval. This extension is essential whenever
biological state cannot be represented adequately by chronological age
alone, and Lefkovitch-type models are now widely used in animal and
plant demography, ecology and conservation biology
\cite{
Caswell,
Tuljapurkar1990,
MorrisDoak2002,
StubbenMilligan2007,
SalgueroGomez2015COMPADRE,
SalgueroGomez2016COMADRE,
Shefferson2021,
Elaydi2024}.

The matrices considered in this paper have the form
\begin{equation}
\label{eq:Lefkovitch}
L=
\begin{pmatrix}
m_1&m_2&\cdots&m_d\\
b_1&c_2&\cdots&0\\
0&b_2&\ddots&\vdots\\
\vdots&&\ddots&c_d
\end{pmatrix},
\end{equation}
where \(d\) is the number of demographic stages,
\(m_j\geq0\) is the fertility contribution of stage \(j\) to the first
stage, \(b_j>0\), \(j=1,\ldots,d-1\), is the progression coefficient
from stage \(j\) to stage \(j+1\), and \(c_j\geq0\),
\(j=2,\ldots,d\), is the retention coefficient in stage \(j\).
The classical Leslie model is recovered when
\[
c_2=\cdots=c_d=0,
\]
whereas an open-group Leslie matrix is obtained when retention occurs
only in the terminal stage,
\[
c_2=\cdots=c_{d-1}=0,
\qquad
c_d>0.
\]
For the empirical classification below, \emph{pure Lefkovitch}
denotes matrices of the form \eqref{eq:Lefkovitch} with positive
retention in at least one internal stage; the three classes are
therefore disjoint.

The corresponding life-cycle graphs are shown schematically in
Figure~\ref{fig:LefkovitchOpenGroupGraphs}.

\begin{figure}[ht]
\centering

\begin{minipage}{0.48\textwidth}
\centering

\resizebox{\textwidth}{!}{%
\begin{tikzpicture}[
    state/.style={
        circle,
        draw,
        minimum size=8mm,
        inner sep=1pt
    },
    transition/.style={
        ->,
        thick
    },
    every node/.style={font=\small}
]

\node[state] (x1) at (0,0) {$1$};
\node[state] (x2) at (1.8,0) {$2$};
\node        (dots) at (3.4,0) {$\cdots$};
\node[state] (xdm1) at (5.0,0) {$d-1$};
\node[state] (xd) at (6.8,0) {$d$};

\draw[transition]
    (x1) -- node[above] {$b_1$} (x2);

\draw[transition]
    (x2) -- node[above] {$b_2$} (dots);

\draw[transition]
    (dots) -- node[above] {$b_{d-2}$} (xdm1);

\draw[transition]
    (xdm1) -- node[above] {$b_{d-1}$} (xd);

\draw[transition]
    (x2) edge[loop above] node[above] {$c_2$} (x2);

\draw[transition]
    (xdm1) edge[loop above] node[above] {$c_{d-1}$} (xdm1);

\draw[transition]
    (xd) edge[loop right] node[right] {$c_d$} (xd);

\draw[transition]
    (x2.south)
    to[out=-120,in=-60,looseness=1.0]
    node[below] {$m_2$}
    (x1.south);

\draw[transition]
    (xdm1.south)
    to[out=-105,in=-75,looseness=1.15]
    node[pos=.55,below] {$m_{d-1}$}
    (x1.south);

\draw[transition]
    (xd.south)
    to[out=-100,in=-80,looseness=1.32]
    node[pos=.63,below] {$m_d$}
    (x1.south);

\draw[transition]
    (x1) edge[loop left] node[left] {$m_1$} (x1);

\end{tikzpicture}
}

\vspace{1mm}
\textbf{(a) Lefkovitch}

\end{minipage}
\hfill
\begin{minipage}{0.48\textwidth}
\centering

\resizebox{\textwidth}{!}{%
\begin{tikzpicture}[
    state/.style={
        circle,
        draw,
        minimum size=8mm,
        inner sep=1pt
    },
    transition/.style={
        ->,
        thick
    },
    every node/.style={font=\small}
]

\node[state] (x1) at (0,0) {$1$};
\node[state] (x2) at (1.8,0) {$2$};
\node        (dots) at (3.4,0) {$\cdots$};
\node[state] (xdm1) at (5.0,0) {$d-1$};
\node[state] (xd) at (6.8,0) {$d$};

\draw[transition]
    (x1) -- node[above] {$b_1$} (x2);

\draw[transition]
    (x2) -- node[above] {$b_2$} (dots);

\draw[transition]
    (dots) -- node[above] {$b_{d-2}$} (xdm1);

\draw[transition]
    (xdm1) -- node[above] {$b_{d-1}$} (xd);

\draw[transition]
    (xd) edge[loop right] node[right] {$\rho$} (xd);

\draw[transition]
    (x2.south)
    to[out=-120,in=-60,looseness=1.0]
    node[below] {$m_2$}
    (x1.south);

\draw[transition]
    (xdm1.south)
    to[out=-105,in=-75,looseness=1.15]
    node[pos=.55,below] {$m_{d-1}$}
    (x1.south);

\draw[transition]
    (xd.south)
    to[out=-100,in=-80,looseness=1.32]
    node[pos=.63,below] {$m_d$}
    (x1.south);

\draw[transition]
    (x1) edge[loop left] node[left] {$m_1$} (x1);

\end{tikzpicture}
}

\vspace{1mm}
\textbf{(b) Open-group Leslie}

\end{minipage}

\caption{Schematic life-cycle graphs for the demographic structures
considered in this paper. 
(a) A Lefkovitch matrix allows fertility transitions \(m_j\),
progression coefficients \(b_j\), and retention coefficients \(c_j\)
in internal and terminal stages.
(b) In an open-group Leslie matrix, retention is absent from the
internal stages and occurs only in the terminal stage, with
\(c_d=\rho>0\). Only representative fertility and retention links are
shown.}
\label{fig:LefkovitchOpenGroupGraphs}

\end{figure}
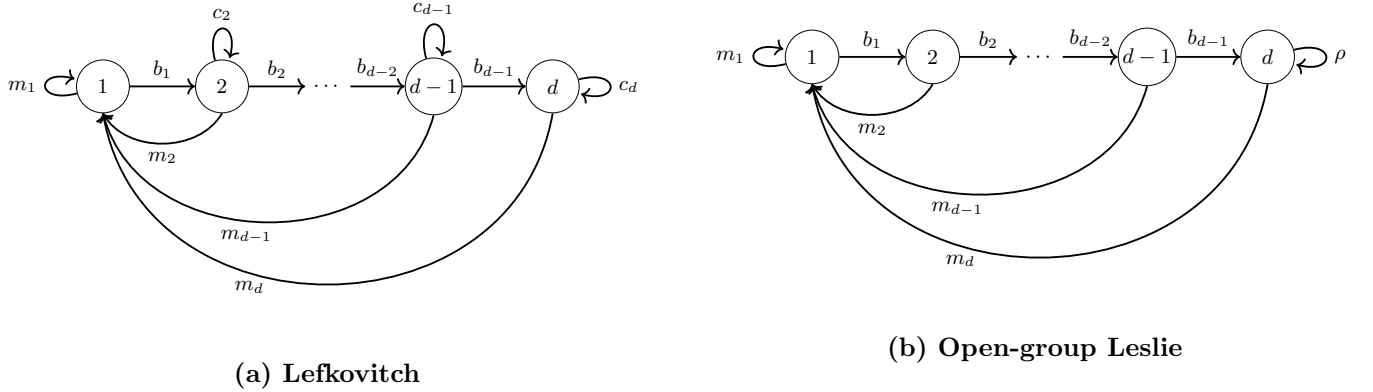

The biological relevance of these structures is apparent in large
empirical demographic databases. Using the versions of the COMADRE
Animal Matrix Database \cite{SalgueroGomez2016COMADRE} and the COMPADRE
Plant Matrix Database \cite{SalgueroGomez2015COMPADRE} accessed on
29 May 2026, we classified the available projection matrices according
to the structures defined above. Among the 3211 COMADRE matrices containing at least one positive
fertility entry, 88.01\% belong to one of the three classes considered
here, and open-group Leslie matrices alone account for 58.46\%. In COMPADRE, the corresponding proportion is
\(18.82\%\), with pure Lefkovitch matrices forming the largest of the
three classes. The complete classification is given in
Table~\ref{tab:databaseprevalence}.

\begin{table}[ht]
\centering
\footnotesize
\caption{Structural classification of projection matrices in COMADRE
and COMPADRE, accessed on 29 May 2026. Percentages refer to matrices containing at least one positive fertility entry.}
\label{tab:databaseprevalence}
\setlength{\tabcolsep}{4.5pt}
\setlength{\extrarowheight}{3pt}
\begin{tabular}{|l|r|r|r|r|r|}
\hline
\textbf{Database}
&
\textbf{Usable}
&
\textbf{Pure Leslie}
&
\textbf{Open-group Leslie}
&
\textbf{Pure Lefkovitch}
&
\textbf{Other}
\tabularnewline
\hline
COMADRE
&
3211
&
651 (20.27\%)
&
1877 (58.46\%)
&
298 (9.28\%)
&
385 (11.99\%)
\tabularnewline
\hline
COMPADRE
&
8163
&
149 (1.83\%)
&
531 (6.50\%)
&
856 (10.49\%)
&
6627 (81.18\%)
\tabularnewline
\hline
\end{tabular}
\end{table}

The contrast between the two databases is also informative. Leslie
and especially open-group Leslie structures are prominent in animal
demography, whereas plant demography is dominated by more general
stage-transition structures. The theory developed here therefore
concerns the Lefkovitch structure \eqref{eq:Lefkovitch} and its Leslie
specialisations; it is not intended as a theory for arbitrary
stage-structured projection matrices.

Let \(\lambda>0\) denote the dominant eigenvalue of \(L\). For a unit
projection interval, the long-term population growth rate is the
Malthusian parameter
\[
r=\log\lambda.
\]
Population growth, however, is only one component of the demographic
description developed by Demetrius. His statistical and bioenergetic
theory relates the Malthusian parameter \(r\) to evolutionary entropy
\(H\) and reproductive potential \(\Phi\)
\cite{Demetrius74,Demetrius75,Demetrius78}. In the classical
age-structured setting,
\[
H=\frac{S}{T},
\]
where \(S\) represents demographic diversity over the reproductive
life cycle and \(T\) is generation time. In a Lefkovitch population,
persistence within stages introduces an additional contribution to
this demographic diversity, which is absent from the classical Leslie
model.

The three quantities are linked by the fundamental relation
\[
r=H+\Phi.
\]
In the bioenergetic interpretation developed by Demetrius, \(r\)
represents the net demographic rate associated with the conversion of
environmental resources into biological work, \(H\) the component
associated with demographic diversity over the life cycle, and
\(\Phi\) the complementary contribution of the resource environment
to demographic growth
\cite{Demetrius97,Demetrius2000}.
Thus \(r\), \(H\) and \(\Phi\) should be regarded as intrinsically
related demographic quantities rather than as independent population
descriptors.

Sensitivity analysis of population growth has a long history.
Demetrius \cite{Demetrius69} derived perturbation expressions for the
dominant eigenvalue of Leslie matrices, and sensitivity and elasticity
analysis subsequently became a standard component of matrix population
theory
\cite{Caswell,Caswell2019}. Demetrius, Gundlach and Ziehe
\cite{Demetrius07} later developed explicit sensitivity formulae for
evolutionary entropy in the classical age-structured Leslie model.
Their work established the natural starting point for the present
problem, but allowing individuals to remain within a stage changes
both the associated Markov chain and the demographic coefficients that
can vary.

Related work by Steiner and Tuljapurkar \cite{SteinerTuljapurkar2020}
derived sensitivities of population entropy for stage-structured Markov
chains with respect to their transition probabilities. Here we derive
explicit sensitivities directly for the fertility, progression, and
retention coefficients of the Lefkovitch projection matrix. We also
consider changes that create new demographic links.

The present paper develops a joint sensitivity theory for
\[
r,\ H,\ \Phi
\]
in irreducible Lefkovitch matrices and in their open-group Leslie
specialisation. Here and throughout, irreducibility refers to the
reproductive component of the life cycle; the full projection matrix
may additionally contain downstream post-reproductive stages. For
\[
\theta\in
\{m_1,\ldots,m_d,\,
  b_1,\ldots,b_{d-1},\,
  c_2,\ldots,c_d\},
\]
we determine
\[
\frac{\partial r}{\partial\theta},
\qquad
\frac{\partial H}{\partial\theta},
\qquad
\frac{\partial\Phi}{\partial\theta}.
\]
The sensitivity of \(r\) follows from the standard formula for a
simple dominant eigenvalue. The main mathematical task is therefore to
derive the response of \(H\) for a Lefkovitch matrix; the response of
reproductive potential then follows from
\[
\frac{\partial\Phi}{\partial\theta}
=
\frac{\partial r}{\partial\theta}
-
\frac{\partial H}{\partial\theta}.
\]
This identity links all three sensitivity calculations below.

Some fertility or retention coefficients may be zero. We must
therefore distinguish between changing a link that is already present
and making a previously absent link positive.

A matrix rescaling result of Alves and Oliveira
\cite{AlvesOliveira2015} provides an additional reason to work with
the original Lefkovitch matrix. Their result converts any irreducible
Lefkovitch matrix into a pseudo-Leslie matrix with the same
eigenvalues, and therefore the same \(r\). However, the converted
matrix may contain negative pseudo-fertilities and need not describe a
biologically admissible population. It also need not preserve the
Markov chain used to define \(H\) and \(\Phi\). Their formulas and
sensitivities must therefore be derived from the original demographic
matrix.

The relation with the classical Leslie sensitivity theory is
summarised in Table~\ref{tab:DGZcomparison}.

The paper makes four main contributions. First, it gives explicit
formulas for \(H\) and \(\Phi\) and develops exact
transition--retention decompositions, with distinct interpretations
for evolutionary entropy and reproductive potential.
Second, it gives closed-form sensitivities of \(r\), \(H\), and
\(\Phi\) to fertility, progression, and retention. Third, it shows why
creating a new demographic link differs from changing an existing one:
the new link produces opposite logarithmic terms in \(H\) and
\(\Phi\), although \(r\) still has a finite derivative. Finally, the
general formulas recover the classical Leslie model and give explicit
expressions for open-group Leslie matrices.

The paper is organised as follows. Section~\ref{sec:2} constructs the
Markov chain associated with an irreducible Lefkovitch matrix and
derives explicit formulas for the long-run stage distribution,
generation time, evolutionary entropy, and reproductive potential. It
also gives the open-group Leslie case. Section~\ref{sec:3} separates
both \(H\) and \(\Phi\) into transition and retention contributions
and recovers the Leslie model when retention vanishes.
Section~\ref{sec:4} derives the sensitivities of \(r\), \(H\), and
\(\Phi\), treating existing and newly created demographic links
separately. Section~\ref{sec:5} applies the formulas to empirical
matrices, and Section~\ref{sec:6} discusses their biological and
mathematical implications.

To keep the biological argument accessible without sacrificing
mathematical precision, the main text presents the results as a
continuous demographic narrative centred on the numbered formulas.
Appendix~\ref{app:SensitivityProofs} collects the corresponding formal
statements---theorems, propositions, corollaries, and lemmas---together
with their proofs. References are provided in both directions: each
result in the main text points to its precise appendix subsection, and
each formal statement in the appendix identifies the equations and the
subsection in which the result is introduced and interpreted.

\begin{table}[H]
\centering
\footnotesize
\caption{Comparison between the sensitivity theory of
Demetrius--Gundlach--Ziehe \cite{Demetrius07} for classical Leslie
matrices and the Lefkovitch theory developed in the present paper.}
\label{tab:DGZcomparison}
\setlength{\tabcolsep}{5pt}
\setlength{\extrarowheight}{5pt}

\begin{tabular}{|p{3.1cm}|p{4.2cm}|p{4.8cm}|}
\hline
\textbf{Feature}
&
\centering\textbf{DGZ \cite{Demetrius07}}
&
\centering\textbf{Present paper}
\tabularnewline
\hline

Population model
&
Classical age-structured Leslie matrix; no stage retention
&
Lefkovitch matrix with explicit stage retention; open-group Leslie
matrices arise as a special case
\tabularnewline
\hline

Demographic variables
&
Fertility \(m_j\) and survival \(b_j\); survivorship
\(\ell_j=b_1\cdots b_{j-1}\) and reproductive contribution
\(V_j=\ell_jm_j\)
&
Fertility \(m_j\), progression \(b_j\), and retention \(c_j\)
\tabularnewline
\hline

Sensitivity targets
&
Malthusian parameter \(r\) and evolutionary entropy \(H\)
&
Malthusian parameter \(r\), evolutionary entropy \(H\), and
reproductive potential \(\Phi\)
\tabularnewline
\hline

Perturbation variables
&
\(m_j,\ b_j,\ \ell_j,\ V_j\)
&
\(m_j,\ b_j,\ c_j\)
\tabularnewline
\hline

Stage-retention sensitivities
&
Absent
&
\(
\partial r/\partial c_j,\quad
\partial H/\partial c_j,\quad
\partial\Phi/\partial c_j
\)
\tabularnewline
\hline

Leslie and open-group cases
&
Classical Leslie model
&
Classical Leslie:
\(c_2=\cdots=c_d=0\);
open-group Leslie:
\(c_2=\cdots=c_{d-1}=0,\ c_d>0\)
\tabularnewline
\hline

\end{tabular}
\end{table}

\endgroup

\section{Evolutionary entropy and reproductive potential for Lefkovitch matrices}\label{sec:2}

\begingroup
\color{black}

Let \(L\) be the irreducible Lefkovitch matrix%
\footnote{\textcolor{black}{The irreducibility assumption concerns the
part of the life cycle that supports reproduction. The full projection
matrix may still contain post-reproductive stages downstream of this
part, as can occur in species with menopause. If those stages have a
lower growth rate than the reproductive part, their long-run
probability is zero. The quantities \(r\), \(H\), and \(\Phi\), and
their sensitivities to changes that leave this downstream structure
unchanged, are then determined by the reproductive part
\cite{StottEtAl2010}. We do not consider life cycles containing two or
more separate groups that can each sustain reproduction.}}%
defined in \eqref{eq:Lefkovitch}, and denote by \(\lambda>0\) its
dominant eigenvalue. Let
\[
u=(u_1,\ldots,u_d)^T
\]
be a corresponding positive right eigenvector. For a unit
projection interval, the associated Malthusian parameter is
\[
r=\log\lambda.
\]

Throughout the paper, empty sums are interpreted as zero, empty
products as \(1\), and terms of the form \(0\log0\) are interpreted
as zero, by continuity.

Our objective in this section is to translate the entries of \(L\)
into quantities with a direct demographic meaning. We first calculate
the relative contribution of reproduction in each stage. We then
rescale the projection matrix to obtain an associated Markov chain.
Its long-run stage distribution supplies the weights used to calculate
\(H\) and \(\Phi\), while \(r\) comes directly from \(\lambda\).

\endgroup

\subsection{From demographic stages to the distribution of reproduction}
\label{subsec:reproductive-probabilities}

The first step is to follow an individual from the first stage to each
possible reproductive stage. Rather than carrying the arbitrary scale
of the eigenvector through the calculation, we describe the
relative contribution of each stage.
\begin{samepage}
For stage \(j\), the required coefficient is
\begin{equation}
R_1=1,
\qquad
R_j=
\prod_{h=1}^{j-1}
\frac{b_h}{\lambda-c_{h+1}},
\qquad
j=2,\ldots,d.
\label{RjDefinition}
\end{equation}
\end{samepage}
The product \(R_j\) accumulates the effects of progression and
retention on the way to stage \(j\). The right eigenvector can then be
written stage by stage as
\begin{equation}
u_j=u_1R_j,
\qquad
j=1,\ldots,d.
\label{PerronRepresentation}
\end{equation}
This makes it natural to assign the reproductive weight
\begin{equation}
q_j=
\frac{m_jR_j}{\lambda},
\qquad
j=1,\ldots,d,
\label{qjDefinition}
\end{equation}
to stage \(j\). These weights are non-negative and, crucially, they
satisfy
\begin{equation}
\sum_{j=1}^{d}q_j=1.
\label{qjNormalization}
\end{equation}
Thus they form a probability distribution. Written directly in terms
of the demographic coefficients, their normalisation is
\begin{equation}
\sum_{j=1}^{d}
\frac{m_jR_j}{\lambda}
=
1.
\label{GeneralizedEulerLotka}
\end{equation}
Equation~\eqref{GeneralizedEulerLotka} is the generalised
Euler--Lotka equation for an irreducible Lefkovitch matrix. It plays
the same role as the familiar age-structured equation: it balances all
stages that contribute to reproduction after accounting for population
growth. The corresponding result and derivation of
equations~\eqref{PerronRepresentation}--
\eqref{GeneralizedEulerLotka} are given in
Appendix~\ref{Proof:EulerLotka}, the subsection devoted specifically
to this result.

\begingroup
\color{black}

The probabilities \(q_j\) describe the distribution of reproductive
transitions generated by \(L\). In the Leslie case they reduce to the
classical reproductive distribution used by Demetrius, Gundlach and
Ziehe.

\endgroup

\subsection{How often each stage is occupied}
\label{subsec:stationary-pathway}

The distribution \(q=(q_1,\ldots,q_d)\) tells us where reproduction
occurs along the stage sequence, but not how frequently the growing
population occupies each stage. To obtain those occupancy weights, let
\begin{equation}
U=
\operatorname{diag}(u_1,\ldots,u_d),
\label{DiagonalMatrix}
\end{equation}
and define
\begin{equation}
P=
\frac1\lambda
U^{-1}LU.
\label{MarkovMatrix}
\end{equation}
Each row of \(P\) sums to one, so \(P\) is the transition matrix of a
Markov chain naturally associated with \(L\). This rescaling leaves
the life-cycle links unchanged and converts their long-term weights
into transition probabilities.

For each stage after the first, define the probabilities of
progression and retention by
\begin{equation}
\eta_i=
\frac{\lambda-c_i}{\lambda},
\qquad
\rho_i=
\frac{c_i}{\lambda},
\qquad
i=2,\ldots,d,
\label{EtaRhoDefinition}
\end{equation}
and let
\begin{equation}
Q_i=
\sum_{j=i}^{d}q_j,
\qquad
i=1,\ldots,d.
\label{TailDistribution}
\end{equation}
be the probability that reproduction occurs in stage \(i\) or a later
stage. The stationary distribution
\[
\pi=(\pi_1,\ldots,\pi_d)
\]
of \(P\) is then
\begin{equation}
\pi_1=\frac1T,
\qquad
\pi_i=
\frac{Q_i}{\eta_iT},
\qquad
i=2,\ldots,d,
\label{StationaryDistribution}
\end{equation}
\begin{samepage}
\noindent where the normalising quantity
\begin{equation}
T=
1+
\sum_{i=2}^{d}
\frac{Q_i}{\eta_i},
\label{TDefinition}
\end{equation}
has the interpretation of generation time for this Markov-chain
description. Equation~\eqref{StationaryDistribution} has a simple
reading: occupancy of stage \(i\) increases with the probability
\(Q_i\) of reaching that stage and with the expected time spent there,
represented by \(1/\eta_i\). The balance calculation establishing
equations~\eqref{StationaryDistribution}--\eqref{TDefinition} is given
in Appendix~\ref{proof:StationaryDistribution}, whose theorem
refers back to the definitions in
equations~\eqref{EtaRhoDefinition}--\eqref{TDefinition}.
\end{samepage}

\begingroup
\color{black}

The pair \((P,\pi)\) provides the probabilities and long-run stage
weights needed to calculate both evolutionary entropy and reproductive
potential. The probability of remaining in stage \(i\) is \(\rho_i\),
whereas \(\eta_i=1-\rho_i\) is the complementary probability
associated with progression into that stage.

\endgroup

\subsection{Entropy and reproductive potential}
\label{subsec:entropy-potential}

We can now measure two sources of demographic uncertainty: the stage
at which reproduction occurs and, within each later stage, whether an
individual progresses or remains there. The evolutionary entropy of
\(L\) is
\begin{equation}
H=
\frac{S}{T},
\label{EntropyRepresentation}
\end{equation}
where its numerator is
\begin{equation}
S=
-\sum_{j=1}^{d}
q_j\log q_j
-
\sum_{i=2}^{d}
\frac{Q_i}{\eta_i}
\left(
\eta_i\log\eta_i
+
\rho_i\log\rho_i
\right),
\label{EntropyS}
\end{equation}
and \(T\) is given by \eqref{TDefinition}. The first sum in
\eqref{EntropyS} measures uncertainty in the stage at which
reproduction occurs. The second measures, stage by stage, the
additional uncertainty between progressing and remaining in the same
stage. Division by \(T\) gives the uncertainty per unit of generation
time. A derivation from the entropy rate of the Markov chain is
given in Appendix~\ref{Proof:EntropyFormula}; the theorem there refers
explicitly to equations~\eqref{EntropyRepresentation}--
\eqref{EntropyS} and to the present subsection.

\begingroup
\color{black}

Using the same Markov chain, reproductive potential is defined by
\begin{equation}
\Phi=
\sum_{\substack{1\leq i,j\leq d\\L_{ij}>0}}
\pi_iP_{ij}\log L_{ij}.
\label{PotentialDefinition}
\end{equation}
For the Lefkovitch matrix \eqref{eq:Lefkovitch}, the non-zero
transition probabilities of \(P\) are
\[
P_{1j}=q_j
\quad\text{when }m_j>0,
\qquad
P_{i,i-1}=\eta_i,
\qquad
P_{ii}=\rho_i
\quad\text{when }c_i>0,
\]
for \(i=2,\ldots,d\). Each positive demographic link is therefore
weighted by the stationary frequency with which it is used and by the
logarithm of its coefficient. This gives \(\Phi\) directly in terms of
the demographic coefficients.

For the irreducible Lefkovitch matrix \(L\), this weighted average is
\begin{equation}
\Phi
=
\pi_1
\sum_{\substack{1\leq j\leq d\\m_j>0}}
q_j\log m_j
+
\sum_{i=2}^{d}
\pi_i\eta_i\log b_{i-1}
+
\sum_{\substack{2\leq i\leq d\\c_i>0}}
\pi_i\rho_i\log c_i .
\label{PotentialRepresentation}
\end{equation}
Using \eqref{StationaryDistribution}, this becomes
\begin{equation}
\Phi=
\frac1T
\left[
\sum_{\substack{1\leq j\leq d\\m_j>0}}
q_j\log m_j
+
\sum_{i=2}^{d}
Q_i\log b_{i-1}
+
\sum_{\substack{2\leq i\leq d\\c_i>0}}
\frac{Q_i\rho_i}{\eta_i}\log c_i
\right].
\label{PotentialExplicit}
\end{equation}
Terms corresponding to \(m_j=0\) or \(c_i=0\) may equivalently be
included by continuous extension.

\eqref{PotentialRepresentation} follows immediately from
\eqref{PotentialDefinition} and the explicit entries of \(P\), while
\eqref{PotentialExplicit} follows by substitution of
\eqref{StationaryDistribution}. The corresponding result, including
its relation to the entropy formula, is stated in
Appendix~\ref{app:PotentialFormula}.

The two demographic quantities obtained from the associated Markov
chain satisfy the fundamental Demetrius identity
\begin{equation}
r=\log\lambda=H+\Phi.
\label{GrowthEntropyPotentialIdentity}
\end{equation}
Indeed, for every positive entry \(L_{ij}\),
\[
P_{ij}
=
\frac{L_{ij}u_j}{\lambda u_i},
\]
and therefore
\[
\log L_{ij}-\log P_{ij}
=
\log\lambda+\log u_i-\log u_j.
\]
The weights \(\pi_iP_{ij}\) give the long-run frequency of each
transition. Averaging with these weights makes the two eigenvector
terms cancel, giving \eqref{GrowthEntropyPotentialIdentity}.

Thus \(r\), \(H\), and \(\Phi\) are calculated from the same
demographic matrix. The quantity \(r=\log\lambda\) gives long-term
growth, while \(H\) and \(\Phi\) describe complementary aspects of
the life cycle. Equations
\eqref{EntropyS} and \eqref{PotentialExplicit} also show explicitly
that stage retention contributes to both quantities. Their
transition--retention decompositions are developed in the next
section.

\endgroup

\subsection{The open-group Leslie case}
\label{subsec:open-group-representation}

\begingroup
\color{black}

The general expressions become especially transparent for an
open-group Leslie matrix. Consider the specialisation of
\eqref{eq:Lefkovitch}, obtained by setting
\[
c_2=\cdots=c_{d-1}=0,
\qquad
c_d=\rho>0.
\]
Only the terminal class can now retain individuals. For every internal
stage,

\endgroup

\[
\eta_i=1,
\qquad
\rho_i=0,
\qquad
i=2,\ldots,d-1,
\]
and
\[
\eta_d=\frac{\lambda-\rho}{\lambda},
\qquad
\rho_d=\frac{\rho}{\lambda}.
\]

\begingroup
\color{black}

Consequently, all internal retention terms vanish and the numerator
\(S\) becomes

\endgroup

\begin{equation}
S=
-\sum_{j=1}^{d}q_j\log q_j
-q_d\log\!\left(\frac{\lambda-\rho}{\lambda}\right)
-\frac{q_d\rho}{\lambda-\rho}
\log\!\left(\frac{\rho}{\lambda}\right),
\label{OpenGroupS}
\end{equation}

\begingroup
\color{black}

The generation time is

\endgroup

\begin{equation}
T=
1+\sum_{i=2}^{d-1}Q_i
+\frac{\lambda q_d}{\lambda-\rho}.
\label{OpenGroupT}
\end{equation}

\begingroup
\color{black}

Dividing \(S\) by this generation time gives

\endgroup

\begin{equation}
H=
\frac{
-\displaystyle\sum_{j=1}^{d}q_j\log q_j
-q_d\log\!\left(\frac{\lambda-\rho}{\lambda}\right)
-\dfrac{q_d\rho}{\lambda-\rho}
\log\!\left(\frac{\rho}{\lambda}\right)
}
{T}.
\label{OpenGroupH}
\end{equation}

\begingroup
\color{black}

The same substitution in the potential gives
\begin{equation}
\Phi=
\frac1T
\left[
\sum_{\substack{1\leq j\leq d\\m_j>0}}
q_j\log m_j
+
\sum_{i=2}^{d}
Q_i\log b_{i-1}
+
\frac{q_d\rho}{\lambda-\rho}\log\rho
\right].
\label{OpenGroupPhi}
\end{equation}
Thus the open-group Leslie formulas follow directly from the general
Lefkovitch expressions. The terminal term records the
repeated residence of individuals in the open class, and
\[
\log\lambda=H+\Phi
\]
continues to hold without a separate derivation.

The complete open-group case
\eqref{OpenGroupS}--\eqref{OpenGroupPhi} is stated and linked to the
general Lefkovitch formulas in
Appendix~\ref{app:OpenGroupRepresentation}.

\endgroup

\begingroup
\color{black}

\section{Decomposition of evolutionary entropy and reproductive potential}
\label{sec:3}
\endgroup

\begingroup
\color{black}

The formulas in the previous section can be organised according to
two features of the life cycle. The first is the distribution of
reproduction across stages. The second is the uncertainty created by
stage persistence: whenever retention is possible, an individual may
either progress or remain in the same stage. We first separate these
two contributions in evolutionary entropy and then introduce the
corresponding transition--retention decomposition of reproductive
potential.

\endgroup


\subsection{Two sources of demographic uncertainty}
\label{subsec:decomposition}

Equation~\eqref{EntropyS} contains two biologically distinct sources
of uncertainty. The first concerns the stage at which reproduction
occurs. The second arises when stage persistence makes both progression
and retention possible within the same stage. Keeping these two sources
visible, evolutionary entropy may be written as
\begin{equation}
H=
H_{\mathrm{tr}}
+
H_{\mathrm{ret}},
\label{EntropyDecomposition}
\end{equation}
where
\begin{equation}
H_{\mathrm{tr}}
=
-\pi_1
\sum_{j=1}^{d}
q_j\log q_j,
\label{ReproductiveEntropy}
\end{equation}
and
\begin{equation}
H_{\mathrm{ret}}
=
-\sum_{i=2}^{d}
\pi_i
\left(
\eta_i\log\eta_i
+
\rho_i\log\rho_i
\right).
\label{RetentionEntropy}
\end{equation}

\begingroup
\color{black}

The decomposition follows directly from the Markov-chain expression
for entropy. The first row of the associated Markov matrix contains
the reproductive transitions, with probabilities \(q_j\), and
contributes \(H_{\mathrm{tr}}\). For each subsequent stage, the
associated row contains the two alternatives represented by
\(\eta_i\) and \(\rho_i\), and these rows together contribute
\(H_{\mathrm{ret}}\). The same decomposition can be read directly
from the derivation of
equations~\eqref{EntropyRepresentation}--\eqref{EntropyS} in
Appendix~\ref{Proof:EntropyFormula}. The corresponding theorem and the
short grouping argument are given in
Appendix~\ref{app:EntropyDecomposition}.

The term \(H_{\mathrm{tr}}\) is therefore the contribution associated
with uncertainty in the distribution \(q_j\) of reproductive
transitions across stages, weighted by the long-run probability
\(\pi_1\) of the first stage. In the Leslie case it reduces to the
classical evolutionary-entropy contribution considered by Demetrius,
Gundlach and Ziehe.

The interpretation of \(H_{\mathrm{ret}}\) requires some care.
It is not the entropy contribution of the retention self-loops alone.
Rather, it measures the additional uncertainty generated when stage
persistence creates a choice between progression and retention. For
each \(i\geq2\), the pair \((\eta_i,\rho_i)\) gives the two
corresponding probabilities in the associated Markov chain. If
\(c_i=0\), then \(\rho_i=0\) and \(\eta_i=1\); progression is then
deterministic at that stage and the entire contribution
\[
-\pi_i
\left(
\eta_i\log\eta_i+\rho_i\log\rho_i
\right)
\]
vanishes. Thus \(H_{\mathrm{ret}}\) is precisely the entropy generated
by the possibility of remaining within a stage, even though its local
binary entropy contains both the progression and retention
probabilities.

This decomposition groups terms according to their demographic
origin; it does not make the two components independent. Both
\(H_{\mathrm{tr}}\) and \(H_{\mathrm{ret}}\) depend, through
\(\lambda\), \(q_j\), \(\eta_i\), \(\rho_i\), and \(\pi\), on the
demographic matrix as a whole.

Reproductive potential admits a related, but not identical,
decomposition. In this case the separation is directly between
demographic coefficients associated with transitions through the life
cycle---fertility and progression---and coefficients associated with
retention self-loops.

Accordingly,
\begin{equation}
\Phi=
\Phi_{\mathrm{tr}}
+
\Phi_{\mathrm{ret}},
\label{PotentialDecomposition}
\end{equation}
where
\begin{equation}
\Phi_{\mathrm{tr}}
=
\pi_1
\sum_{\substack{1\leq j\leq d\\m_j>0}}
q_j\log m_j
+
\sum_{i=2}^{d}
\pi_i\eta_i\log b_{i-1},
\label{TransitionPotential}
\end{equation}
and
\begin{equation}
\Phi_{\mathrm{ret}}
=
\sum_{\substack{2\leq i\leq d\\c_i>0}}
\pi_i\rho_i\log c_i.
\label{RetentionPotential}
\end{equation}

Equations~\eqref{PotentialDecomposition}--\eqref{RetentionPotential}
follow immediately from \eqref{PotentialRepresentation} by separating
the fertility and progression entries from the retention entries.
Thus \(\Phi_{\mathrm{tr}}\) contains the contributions of demographic
links that generate reproduction or progression, whereas
\(\Phi_{\mathrm{ret}}\) contains precisely the contributions of
retention self-loops. The corresponding proposition and its short
proof are given in Appendix~\ref{app:PotentialDecomposition}.

The notation ``transition'' and ``retention'' therefore has slightly
different roles in the two decompositions. For \(H\),
\(H_{\mathrm{tr}}\) is the entropy of the reproductive-transition
distribution, while \(H_{\mathrm{ret}}\) is the binary uncertainty
between progression and retention created by stage persistence. For
\(\Phi\), by contrast, \(\Phi_{\mathrm{tr}}\) and
\(\Phi_{\mathrm{ret}}\) are obtained by an explicit partition of the
demographic links themselves. In neither case are the resulting
components independent functions of separate sets of demographic
parameters.

Combining equations~\eqref{EntropyDecomposition} and
\eqref{PotentialDecomposition} with
\eqref{GrowthEntropyPotentialIdentity} gives the exact identity
\begin{equation}
r
=
\bigl(
H_{\mathrm{tr}}+\Phi_{\mathrm{tr}}
\bigr)
+
\bigl(
H_{\mathrm{ret}}+\Phi_{\mathrm{ret}}
\bigr).
\label{GrowthDecomposition}
\end{equation}
Equation~\eqref{GrowthDecomposition} is an exact regrouping of the
fundamental identity \(r=H+\Phi\). It does not define separate
transition and retention growth rates, nor does it imply that the two
parenthesised terms are dynamically independent. Its purpose is only
to keep visible the two structural features introduced above.

\endgroup


\begingroup
\color{black}

\subsection{Recovering the age-structured model}
\label{subsec:leslie-limit}

The connection with the classical Leslie model is obtained by removing
stage persistence. Consider the Lefkovitch matrix while keeping its
fertility coefficients \(m_j\) and progression coefficients \(b_j\)
fixed and letting
\[
c_i\to0,
\qquad
i=2,\ldots,d.
\]
Let \(\lambda_L\) denote the dominant eigenvalue of the resulting
Leslie matrix. Throughout this subsection, the superscript \((0)\)
denotes the corresponding limiting value at zero retention. Thus,
using equations~\eqref{RjDefinition}, \eqref{qjDefinition}, and
\eqref{TailDistribution},
\[
q_j^{(0)}
=
\frac{m_j}{\lambda_L}
\prod_{h=1}^{j-1}
\frac{b_h}{\lambda_L}
=
\frac{m_j b_1\cdots b_{j-1}}{\lambda_L^{\,j}},
\qquad
Q_i^{(0)}
=
\sum_{j=i}^{d}q_j^{(0)},
\]
with the empty product interpreted as \(1\). From
\eqref{StationaryDistribution}, the corresponding generation time is
\[
T_L
=
1+\sum_{i=2}^{d}Q_i^{(0)}.
\]

As \(c_i\to0\), one has
\[
\rho_i\to0,
\qquad
\eta_i\to1,
\]
so the uncertainty associated with progression versus retention
disappears. Consequently,
\[
H_{\mathrm{ret}}\longrightarrow0.
\]
At the same time, every retention self-loop disappears from
reproductive potential, and
\[
\Phi_{\mathrm{ret}}\longrightarrow0.
\]

The Lefkovitch entropy therefore converges to
\begin{equation}
H_L=
-\frac1{T_L}
\sum_{j=1}^{d}
q_j^{(0)}
\log q_j^{(0)},
\label{LeslieEntropyFormula}
\end{equation}
which is precisely the evolutionary entropy of the corresponding
Leslie model.

The reproductive potential converges simultaneously to
\begin{equation}
\Phi_L
=
\frac1{T_L}
\left[
\sum_{\substack{1\leq j\leq d\\m_j>0}}
q_j^{(0)}\log m_j
+
\sum_{i=2}^{d}
Q_i^{(0)}\log b_{i-1}
\right],
\label{LesliePotentialFormula}
\end{equation}
and therefore
\[
r_L=H_L+\Phi_L,
\qquad
r_L=\log\lambda_L.
\]

\endgroup

\begingroup
\color{black}

The entropy limit is proved in Appendix~\ref{Proof:LeslieLimit}. The
corresponding result for reproductive potential follows directly from
\eqref{PotentialExplicit}. Indeed, as \(c_i\to0\),
\(\rho_i=c_i/\lambda\to0\), while \(\lambda\) converges to the
positive Leslie eigenvalue \(\lambda_L\). Hence
\[
\rho_i\log c_i
=
\frac{c_i}{\lambda}\log c_i
\longrightarrow0,
\]
because \(x\log x\to0\) as \(x\downarrow0\). At the same time,
\(q_j\), \(Q_i\), and \(T\) converge to the limiting quantities
\(q_j^{(0)}\), \(Q_i^{(0)}\), and \(T_L\) defined above. No additional
limiting argument for \(\Phi\) is therefore required.

The Leslie limit removes both effects introduced by stage persistence:
the binary progression--retention uncertainty in
\(H_{\mathrm{ret}}\) and the explicit retention-link contribution
\(\Phi_{\mathrm{ret}}\). What remains is exactly the classical
age-structured description of evolutionary entropy and reproductive
potential.

\endgroup

\section{Joint sensitivity of growth, evolutionary entropy and reproductive potential}
\label{sec:4}

\subsection{A common starting point}
\label{subsec:sensitivity-start}

\begingroup
\color{black}

The identity established in
\eqref{GrowthEntropyPotentialIdentity},
\[
r=H+\Phi,
\qquad
r=\log\lambda,
\]
provides the starting point for the sensitivity calculations.
Here a sensitivity describes the response of one of the three
quantities to a small change in a demographic coefficient, with the
remaining coefficients held fixed. For any such change
for which the three quantities are differentiable,
\begin{equation}
dr=dH+d\Phi.
\label{SensitivityIdentity}
\end{equation}
Thus the sensitivities of \(r\), \(H\), and \(\Phi\) are not
independent. The sensitivity of the Malthusian parameter is classical,
whereas the principal calculation required here is the response of
evolutionary entropy to changes in fertility, progression, and
retention. Once \(dr\) and \(dH\) are known, \(d\Phi\) follows from
\eqref{SensitivityIdentity}.

Let \(v=(v_1,\ldots,v_d)\) be a positive left eigenvector associated
with \(\lambda\), normalised together with the right eigenvector \(u\) by
\[
vu=1.
\]

Consider first the familiar response of population growth. For an
infinitesimal change \(dL\) in an irreducible Lefkovitch matrix,
\begin{equation}
d\lambda=v(dL)u,
\qquad
dr=\frac{1}{\lambda}v(dL)u.
\label{GrowthDifferential}
\end{equation}
Consequently,
\begin{equation}
\frac{\partial r}{\partial m_k}
=
\frac{v_1u_k}{\lambda},
\qquad
\frac{\partial r}{\partial b_k}
=
\frac{v_{k+1}u_k}{\lambda},
\qquad
\frac{\partial r}{\partial c_k}
=
\frac{v_ku_k}{\lambda},
\label{GrowthSensitivityFormulae}
\end{equation}
with the natural index ranges
\[
k=1,\ldots,d
\quad\text{for }m_k,
\qquad
k=1,\ldots,d-1
\quad\text{for }b_k,
\qquad
k=2,\ldots,d
\quad\text{for }c_k.
\]
This is the standard perturbation formula for a simple dominant
eigenvalue. It also has a useful demographic interpretation. In a link
from stage \(j\) to stage \(i\), \(u_j\) measures the long-term
abundance of the source stage and \(v_i\) measures the
reproductive value of the destination. Their product, divided by
\(\lambda\), is the growth response to a small change in that link.
No new perturbation theory is therefore required for \(r\). The
precise theorem corresponding to
equations~\eqref{GrowthDifferential}--\eqref{GrowthSensitivityFormulae}
and their short dominant-eigenvalue proof are included in
Appendix~\ref{app:GrowthSensitivity}.

A distinction that will be important below is already visible here.
Because the dominant eigenvalue is simple, the matrix-entry derivative
in \eqref{GrowthSensitivityFormulae} remains finite when a fertility or
retention coefficient is zero. At such a boundary point the admissible
demographic perturbation is one-sided, but \(r\) still has an ordinary
first-order response. Evolutionary
entropy and reproductive potential behave differently at such boundary
points because their definitions contain logarithmic terms. We therefore
separate changes to links that are already present from perturbations
that create a previously absent link.

\subsection{Changes to existing demographic links}
\label{subsec:support-preserving}

We first consider changes to links that are already present. Define
\[
\mathcal F
=
\{j\in\{1,\ldots,d\}:m_j>0\},
\qquad
\mathcal R
=
\{i\in\{2,\ldots,d\}:c_i>0\}.
\]
An \emph{existing-link perturbation} keeps the set of demographic
links fixed: positive fertility and retention coefficients may vary, while
\(m_j=0\) for \(j\notin\mathcal F\) and
\(c_i=0\) for \(i\notin\mathcal R\) remain zero. The progression
coefficients \(b_i\) are positive by assumption and may therefore be
varied without changing the support. Changes that turn a zero fertility
or retention coefficient into a positive coefficient leave this
setting, create a new demographic link, and are treated
separately in Subsection~\ref{subsec:new-links}.

\subsubsection{First step: changes in reproductive and retention probabilities}
\label{subsubsec:normalized-pathways}

The calculation is easiest to follow in two layers. We first ask how
\(H\) responds when the reproductive and retention
probabilities change. We then translate those changes back into the
original fertility, progression, and retention coefficients.

Define the local retention term
\begin{equation}
h_i=
-\eta_i\log\eta_i
-
\rho_i\log\rho_i,
\qquad
i=2,\ldots,d.
\label{EntropyKernel}
\end{equation}
With the convention \(0\log0=0\), equation
\eqref{EntropyS} may be written as
\begin{equation}
S=
-\sum_{j=1}^{d}
q_j\log q_j
+
\sum_{i=2}^{d}
\frac{Q_i}{\eta_i}h_i.
\label{EntropyFunctional}
\end{equation}

The first step is to express the differential of \(H\) in terms of
the variables \(q_j\) and \(\eta_i\).


For an existing-link perturbation, the first step takes the compact
form
\begin{equation}
dH
=
\sum_{j\in\mathcal F}
C_j\,dq_j
+
\sum_{i\in\mathcal R}
D_i\,d\eta_i,
\label{EntropyDifferential}
\end{equation}
where
\begin{equation}
C_j=
\frac1T
\left[
-(1+\log q_j)
+
\sum_{i=2}^{j}
\frac{h_i-H}{\eta_i}
\right],
\qquad
j\in\mathcal F,
\label{CjDefinition}
\end{equation}
and
\begin{equation}
D_i=
\frac{Q_i}{T\eta_i^2}
\left(
H+\log\rho_i
\right),
\qquad
i\in\mathcal R.
\label{DiDefinition}
\end{equation}

The coefficients \(C_j\) and \(D_i\) are the coefficients of the
probability changes in the total differential \eqref{EntropyDifferential}.
The term \(C_j\,dq_j\) records the contribution produced by changing
the reproductive weight of stage \(j\), whereas \(D_i\,d\eta_i\)
records the contribution produced by changing the balance between
progression and retention at stage \(i\). These are not independent
probability perturbations: the \(q_j\) remain normalised and all the
quantities involved are determined by the same demographic matrix. Their
derivation is given in Appendix~\ref{Lemma1}.

The restriction to \(\mathcal F\) and \(\mathcal R\) is essential.
If \(m_j=0\), then \(q_j=0\) and an existing-link perturbation has
\(dq_j=0\). Likewise, if \(c_i=0\), then
\(\rho_i=0\), \(\eta_i=1\), and \(d\eta_i=0\) throughout an
existing-link perturbation. The zero entries do not change, so the
logarithmic terms generated by the creation of new links do not enter
this differential calculation.

\subsubsection{Second step: returning to demographic coefficients}
\label{subsubsec:parameter-differentials}

To make the preceding expression usable with empirical matrices, we
must express \(dq_j\) and \(d\eta_i\) in terms of the demographic
parameters.

The repeated appearance of \(\lambda\) in the formula for
reproduction at stage \(j\) can be collected in the auxiliary quantity
\begin{equation}
A_j=
\frac1\lambda
+
\sum_{h=1}^{j-1}
\frac1{\lambda-c_{h+1}},
\qquad
j=1,\ldots,d.
\label{AjDefinition}
\end{equation}
For \(j\in\mathcal F\), direct differentiation of the reproductive
weight gives
\begin{equation}
dq_j
=
q_j
\left[
\frac{dm_j}{m_j}
+
\sum_{h=1}^{j-1}
\frac{db_h}{b_h}
+
\sum_{h=1}^{j-1}
\frac{dc_{h+1}}{\lambda-c_{h+1}}
-
A_j\,d\lambda
\right],
\label{dqjFormula}
\end{equation}
where terms \(dc_{h+1}\) corresponding to
\(h+1\notin\mathcal R\) vanish under an existing-link
perturbation. The corresponding change in the progression probability
is
\begin{equation}
d\eta_i
=
\frac{c_i}{\lambda^2}d\lambda
-
\frac1\lambda dc_i,
\qquad
i\in\mathcal R.
\label{detaFormula}
\end{equation}

Equation~\eqref{dqjFormula} shows how each demographic coefficient
affects reproduction. A fertility change acts directly through
\(m_j\); a change in progression or retention affects every later
reproductive stage reached through that link; and \(d\lambda\)
accounts for the change in long-term growth. The details are given in
Appendix~\ref{Lemma2}.

\subsubsection{Assembling the entropy response}
\label{subsubsec:entropy-response}

The two layers can now be assembled into a single expression in the
demographic parameters. The quantity
\begin{equation}
B=
-\sum_{j\in\mathcal F}
C_jq_jA_j
+
\sum_{i\in\mathcal R}
D_i\frac{c_i}{\lambda^2},
\label{BDefinition}
\end{equation}
collects the part of the entropy response caused by the change in the
dominant eigenvalue. With this abbreviation,
\begin{align}
dH
&=
\sum_{j\in\mathcal F}
C_jq_j
\frac{dm_j}{m_j}
+
\sum_{k=1}^{d-1}
\left(
\sum_{\substack{j\in\mathcal F\\j\geq k+1}}
C_jq_j
\right)
\frac{db_k}{b_k}
\nonumber\\
&\quad
+
\sum_{k\in\mathcal R}
\left[
\frac1{\lambda-c_k}
\sum_{\substack{j\in\mathcal F\\j\geq k}}
C_jq_j
-
\frac{D_k}{\lambda}
\right]
dc_k
+
B\,d\lambda .
\label{IntermediateDifferential}
\end{align}

Equation~\eqref{IntermediateDifferential} is the bridge between the
probability-based description and the empirical matrix entries. Its
derivation is given in Appendix~\ref{Theorem1}.

Substituting the dominant-eigenvalue formula
\eqref{GrowthDifferential} eliminates \(d\lambda\) and gives the
coefficientwise sensitivities.

\subsubsection{Coefficientwise sensitivities}
\label{subsubsec:coefficientwise-sensitivities}

For an existing-link perturbation, the total entropy response is
\begin{align}
dH
&=
\sum_{k\in\mathcal F}
\frac{\partial H}{\partial m_k}\,dm_k
+
\sum_{k=1}^{d-1}
\frac{\partial H}{\partial b_k}\,db_k
\nonumber\\
&\quad
+
\sum_{k\in\mathcal R}
\frac{\partial H}{\partial c_k}\,dc_k,
\label{TotalSensitivityFormula}
\end{align}
where
\begin{equation}
\frac{\partial H}{\partial m_k}
=
\frac{C_kq_k}{m_k}
+
Bv_1u_k,
\qquad
k\in\mathcal F,
\label{SensitivityM}
\end{equation}
\begin{equation}
\frac{\partial H}{\partial b_k}
=
\frac1{b_k}
\sum_{\substack{j\in\mathcal F\\j\geq k+1}}
C_jq_j
+
Bv_{k+1}u_k,
\qquad
k=1,\ldots,d-1,
\label{SensitivityB}
\end{equation}
and
\begin{equation}
\frac{\partial H}{\partial c_k}
=
\frac1{\lambda-c_k}
\sum_{\substack{j\in\mathcal F\\j\geq k}}
C_jq_j
-
\frac{D_k}{\lambda}
+
Bv_ku_k,
\qquad
k\in\mathcal R.
\label{SensitivityC}
\end{equation}

In each of equations~\eqref{SensitivityM}--\eqref{SensitivityC}, the
terms containing \(C_j\) or \(D_i\) describe the direct change in
reproductive and retention probabilities. The final term, containing
\(Bv_i u_j\), describes the part caused by the change in long-term
population growth. The algebra leading to these expressions is given
in Appendix~\ref{Theorem2}.

Once the growth and entropy responses are known, reproductive
potential requires no additional differential calculation: combine
equations~\eqref{GrowthSensitivityFormulae} and
\eqref{SensitivityM}--\eqref{SensitivityC} with
\eqref{SensitivityIdentity}.

For an existing-link perturbation, this gives
\begin{equation}
\frac{\partial\Phi}{\partial m_k}
=
-\frac{C_kq_k}{m_k}
+
\left(
\frac1\lambda-B
\right)v_1u_k,
\qquad
k\in\mathcal F,
\label{PotentialSensitivityM}
\end{equation}
\begin{equation}
\frac{\partial\Phi}{\partial b_k}
=
-\frac1{b_k}
\sum_{\substack{j\in\mathcal F\\j\geq k+1}}
C_jq_j
+
\left(
\frac1\lambda-B
\right)v_{k+1}u_k,
\qquad
k=1,\ldots,d-1,
\label{PotentialSensitivityB}
\end{equation}
and
\begin{equation}
\frac{\partial\Phi}{\partial c_k}
=
-\frac1{\lambda-c_k}
\sum_{\substack{j\in\mathcal F\\j\geq k}}
C_jq_j
+
\frac{D_k}{\lambda}
+
\left(
\frac1\lambda-B
\right)v_ku_k,
\qquad
k\in\mathcal R.
\label{PotentialSensitivityC}
\end{equation}

These identities are immediate from
\[
\frac{\partial\Phi}{\partial\theta}
=
\frac{\partial r}{\partial\theta}
-
\frac{\partial H}{\partial\theta},
\]
so they require no further differential calculation. The corresponding
corollary is given in Appendix~\ref{app:PotentialSensitivity}, where it
is linked explicitly to equations~\eqref{SensitivityIdentity},
\eqref{GrowthSensitivityFormulae}, and
\eqref{SensitivityM}--\eqref{SensitivityC}.

The structure of these formulas is worth keeping in view. The entropy
sensitivity coefficients contain two conceptually
different contributions. The terms involving \(C_j\) and \(D_i\)
describe the direct response of reproductive, progression, and
retention probabilities, whereas the terms
\[
Bv_i u_j
\]
arise through the change in the dominant eigenvalue. Reproductive
potential contains the complementary contribution
\[
\left(\frac1\lambda-B\right)v_i u_j.
\]
Thus the exact identity \(dr=dH+d\Phi\) holds parameter by parameter.

\subsection{Creation of new demographic links}
\label{subsec:new-links}

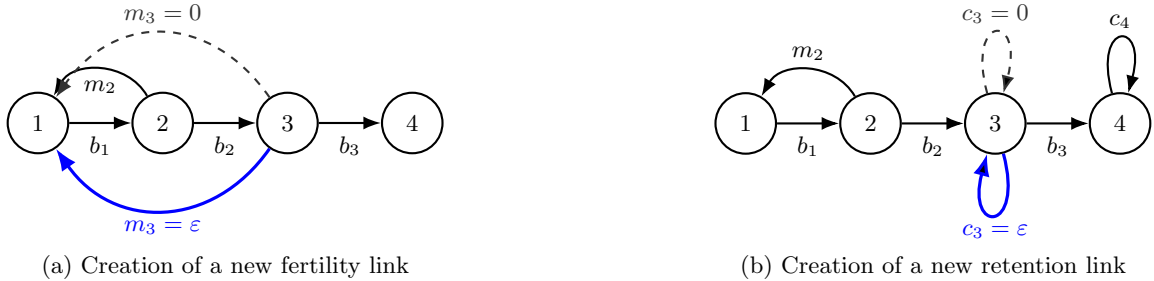
\begin{figure}[!htbp]
\centering

\begin{minipage}{0.48\textwidth}
\centering
\begin{tikzpicture}[
    >=Latex,
    x=1cm,
    y=1cm,
    every node/.style={font=\small},
    stage/.style={
        circle,
        draw=black,
        thick,
        fill=white,
        minimum size=8mm,
        inner sep=0pt
    },
    prog/.style={->,thick,black},
    existing/.style={->,thick,black},
    absent/.style={->,thick,dashed,draw=gray!160},
    newlink/.style={->,very thick,draw=blue},
    lab/.style={fill=white,inner sep=1.5pt}
]

\node[stage] (a1) at (0,0) {$1$};
\node[stage] (a2) at (1.65,0) {$2$};
\node[stage] (a3) at (3.30,0) {$3$};
\node[stage] (a4) at (4.95,0) {$4$};

\draw[prog]
    (a1) -- node[below=3pt,lab] {$b_1$} (a2);

\draw[prog]
    (a2) -- node[below=3pt,lab] {$b_2$} (a3);

\draw[prog]
    (a3) -- node[below=3pt,lab] {$b_3$} (a4);

\draw[existing]
    (a2) to[out=120,in=60,looseness=1.10]
    node[pos=0.50,below=2pt,lab] {$m_2$}
    (a1);

\draw[absent]
    (a3) to[out=125,in=55,looseness=1.25]
    node[pos=0.50,above=2pt,lab,text=gray!160] {$m_3=0$}
    (a1);

\draw[newlink]
    (a3) to[out=235,in=305,looseness=1.20]
    node[pos=0.50,below=2pt,lab,text=blue] {$m_3=\varepsilon$}
    (a1);

\node[font=\small] at (2.48,-1.90)
    {(a) Creation of a new fertility link};

\end{tikzpicture}
\end{minipage}
\hfill
\begin{minipage}{0.48\textwidth}
\centering
\begin{tikzpicture}[
    >=Latex,
    x=1cm,
    y=1cm,
    every node/.style={font=\small},
    stage/.style={
        circle,
        draw=black,
        thick,
        fill=white,
        minimum size=8mm,
        inner sep=0pt
    },
    prog/.style={->,thick,black},
    existing/.style={->,thick,black},
    absent/.style={->,thick,dashed,draw=gray!160},
    newlink/.style={->,very thick,draw=blue},
    lab/.style={fill=white,inner sep=1.5pt}
]

\node[stage] (b1) at (0,0) {$1$};
\node[stage] (b2) at (1.65,0) {$2$};
\node[stage] (b3) at (3.30,0) {$3$};
\node[stage] (b4) at (4.95,0) {$4$};

\draw[prog]
    (b1) -- node[below=3pt,lab] {$b_1$} (b2);

\draw[prog]
    (b2) -- node[below=3pt,lab] {$b_2$} (b3);

\draw[prog]
    (b3) -- node[below=3pt,lab] {$b_3$} (b4);

\draw[existing]
    (b2) to[out=120,in=60,looseness=1.10]
    node[pos=0.50,above=2pt,lab] {$m_2$}
    (b1);

\draw[existing]
    (b4) edge[loop above,min distance=10mm,looseness=7]
    node[above=2pt,lab] {$c_4$}
    (b4);

\draw[absent]
    (b3) edge[loop above,min distance=11mm,looseness=7]
    node[above=2pt,lab,text=gray!160] {$c_3=0$}
    (b3);

\draw[newlink]
    (b3) edge[loop below,min distance=11mm,looseness=7]
    node[below=2pt,lab,text=blue] {$c_3=\varepsilon$}
    (b3);

\node[font=\small] at (2.48,-1.90)
    {(b) Creation of a new retention link};

\end{tikzpicture}
\end{minipage}

\vspace{0.5em}

\caption{
Boundary perturbations created by turning a zero coefficient into a
positive one. Black arrows denote demographic links already present in
the projection matrix. A grey dashed arrow indicates a link absent at
baseline ($\varepsilon=0$), whereas the blue arrow represents the same
link after creation ($\varepsilon>0$). Panel~(a) shows the creation of
a new fertility link, $m_3=\varepsilon$, and panel~(b) the creation of
a new retention link, $c_3=\varepsilon$. These perturbations differ
qualitatively from changes to existing links:
$r(\varepsilon)=r(0)+O(\varepsilon)$, whereas evolutionary entropy and
reproductive potential have opposite leading terms of order
$\varepsilon\log\varepsilon$, namely
$H(\varepsilon)=H(0)-\kappa\varepsilon\log\varepsilon+O(\varepsilon)$
and
$\Phi(\varepsilon)=\Phi(0)+\kappa\varepsilon\log\varepsilon+O(\varepsilon)$,
where \(\kappa>0\) denotes the corresponding coefficient
\(\kappa_k^{(m)}\) or \(\kappa_k^{(c)}\) defined below.
}
\label{fig:new-links}
\end{figure}

The preceding differential formulas describe changes that keep the set
of demographic links fixed. A different situation occurs at the boundary
of this parameter region, when
a coefficient that is initially zero is made positive. Within the
Lefkovitch structure \eqref{eq:Lefkovitch}, this happens when a
previously absent fertility link \(m_k=0\) or retention link
\(c_k=0\) is created. Since the progression coefficients satisfy
\(b_k>0\) by assumption, their perturbations remain within the existing
support.

Figure~\ref{fig:new-links} illustrates the two boundary perturbations
considered in this subsection. Because the dominant eigenvalue is simple,
both the perturbed eigenvalue \(\lambda(\varepsilon)\) and
\(r(\varepsilon)=\log\lambda(\varepsilon)\) vary smoothly as the new
coefficient leaves zero. Evolutionary entropy and reproductive
potential, however, contain logarithmic terms associated with the newly
created probability and generally do not have finite one-sided
derivatives.

Let \(L_0\) denote the baseline irreducible Lefkovitch matrix. Write
\(\lambda_0\), \(T_0\), and \(\pi^{(0)}\) for its dominant
eigenvalue, generation time, and stationary distribution. Throughout
this subsection, a superscript \((0)\) denotes evaluation at
\(L_0\). Thus \(R_j^{(0)}\), \(q_j^{(0)}\), and
\(Q_j^{(0)}\) are the quantities defined in
\eqref{RjDefinition}, \eqref{qjDefinition}, and
\eqref{TailDistribution}, evaluated at the baseline matrix. We write
\(u\) and \(v\) for the corresponding right and left eigenvectors of
\(L_0\), normalised by \(vu=1\).

\paragraph{Creating a fertility link.}
Suppose \(m_k=0\) in \(L_0\) and define \(L_\varepsilon\) by
replacing \(m_k\) with \(\varepsilon>0\), all other entries being
fixed. From \eqref{qjDefinition},
\[
q_k(\varepsilon)
=
\frac{\varepsilon R_k(\varepsilon)}{\lambda(\varepsilon)}
=
\frac{R_k^{(0)}}{\lambda_0}\,\varepsilon
+
O(\varepsilon^2).
\]
Thus the probability carried by the newly created fertility link is
of first order in \(\varepsilon\). As \(\varepsilon\downarrow0\),
\begin{align}
r(\varepsilon)
&=
r(0)
+
\frac{v_1u_k}{\lambda_0}\varepsilon
+
O(\varepsilon^2),
\label{BoundaryGrowthFertility}\\
H(\varepsilon)
&=
H(0)
-
\kappa_k^{(m)}
\varepsilon\log\varepsilon
+
O(\varepsilon),
\label{BoundaryEntropyFertility}\\
\Phi(\varepsilon)
&=
\Phi(0)
+
\kappa_k^{(m)}
\varepsilon\log\varepsilon
+
O(\varepsilon),
\label{BoundaryPotentialFertility}
\end{align}
where
\begin{equation}
\kappa_k^{(m)}
=
\frac{R_k^{(0)}}{\lambda_0T_0}
>0.
\label{BoundaryFertilityCoefficient}
\end{equation}
The coefficient in \eqref{BoundaryFertilityCoefficient} follows directly
from the new entropy term. By
\eqref{StationaryDistribution}, \(\pi_1^{(0)}=1/T_0\), so
\[
-\pi_1(\varepsilon)q_k(\varepsilon)\log q_k(\varepsilon)
=
-\frac{R_k^{(0)}}{\lambda_0T_0}
\,\varepsilon\log\varepsilon
+
O(\varepsilon).
\]
This is the origin of \(\kappa_k^{(m)}\). Since \(r=H+\Phi\) and
\(r\) has an ordinary first-order expansion, the singular contribution
to \(\Phi\) has the opposite sign.

\paragraph{Creating a retention link.}
Suppose \(c_k=0\) in \(L_0\) and define \(L_\varepsilon\) by
replacing \(c_k\) with \(\varepsilon>0\), all other entries being
fixed. From \eqref{EtaRhoDefinition},
\[
\rho_k(\varepsilon)
=
\frac{\varepsilon}{\lambda(\varepsilon)}
=
\frac{\varepsilon}{\lambda_0}
+
O(\varepsilon^2).
\]
At the baseline matrix, \(c_k=0\) implies
\(\rho_k^{(0)}=0\) and \(\eta_k^{(0)}=1\). Therefore
\eqref{StationaryDistribution} gives
\[
\pi_k^{(0)}
=
\frac{Q_k^{(0)}}{T_0}.
\]
These two relations determine the coefficient of the logarithmic term.
Then
\begin{align}
r(\varepsilon)
&=
r(0)
+
\frac{v_ku_k}{\lambda_0}\varepsilon
+
O(\varepsilon^2),
\label{BoundaryGrowthRetention}\\
H(\varepsilon)
&=
H(0)
-
\kappa_k^{(c)}
\varepsilon\log\varepsilon
+
O(\varepsilon),
\label{BoundaryEntropyRetention}\\
\Phi(\varepsilon)
&=
\Phi(0)
+
\kappa_k^{(c)}
\varepsilon\log\varepsilon
+
O(\varepsilon),
\label{BoundaryPotentialRetention}
\end{align}
where
\begin{equation}
\kappa_k^{(c)}
=
\frac{\pi_k^{(0)}}{\lambda_0}
=
\frac{Q_k^{(0)}}{\lambda_0T_0}
>0.
\label{BoundaryRetentionCoefficient}
\end{equation}
Indeed, the new non-regular contribution to the retention part of the
entropy is
\[
-\pi_k(\varepsilon)
\left[
\eta_k(\varepsilon)\log\eta_k(\varepsilon)
+
\rho_k(\varepsilon)\log\rho_k(\varepsilon)
\right]
=
-\frac{\pi_k^{(0)}}{\lambda_0}
\,\varepsilon\log\varepsilon
+
O(\varepsilon),
\]
and the equality \(\pi_k^{(0)}=Q_k^{(0)}/T_0\) gives the second form
in \eqref{BoundaryRetentionCoefficient}. Again, the corresponding
logarithmic term in \(\Phi\) has the opposite sign because
\(r=H+\Phi\).

Consequently, in either case,
\begin{equation}
\frac{dH}{d\varepsilon}\bigg|_{0^+}
=
+\infty,
\qquad
\frac{d\Phi}{d\varepsilon}\bigg|_{0^+}
=
-\infty,
\label{BoundarySingularSensitivity}
\end{equation}
whereas the one-sided derivative of \(r\) is finite.

The paired signs have a clear consequence. A newly created link has a
finite first-order effect on growth, whereas the leading changes in
\(H\) and \(\Phi\) have magnitude \(\varepsilon|\log\varepsilon|\).
Since \(\varepsilon|\log\varepsilon|/\varepsilon\to\infty\) as
\(\varepsilon\downarrow0\), the initial redistribution between entropy
and potential is asymptotically larger than the corresponding
first-order change in growth. The derivation of these leading terms is
given in Appendix~\ref{Proof:BoundaryLinks}.

Equations~\eqref{BoundaryEntropyFertility}--
\eqref{BoundarySingularSensitivity} distinguish continuity from
differentiability as a demographic coefficient leaves zero. The
conventions used in Section~2 make \(H\) and \(\Phi\) continuous when
a fertility or retention coefficient tends to zero, but the creation
of the corresponding link is generally not described by a finite
entropy or potential sensitivity coefficient.

Changes that create matrix entries outside the Lefkovitch
pattern \eqref{eq:Lefkovitch} are not considered here, since they
leave the class of population matrices studied in this paper.

\subsection{Leslie and open-group reductions}
\label{subsec:sensitivity-reductions}

The sensitivity formulas recover the
Demetrius--Gundlach--Ziehe age-structured theory by approaching the
Leslie model from within the Lefkovitch class. More precisely, consider
a family in which the same fertility links remain present and the
retention coefficients are positive before tending to zero:
\[
c_i\to0,
\qquad
i=2,\ldots,d.
\]
Along these existing-link directions, the sensitivities with respect to
fertility and progression converge to the corresponding Leslie
sensitivities. In particular, the evolutionary-entropy sensitivities
reduce to the Demetrius--Gundlach--Ziehe formulas, the growth
sensitivities reduce to the classical Leslie perturbation formulas, and
\[
\frac{\partial\Phi}{\partial\theta}
=
\frac{\partial r}{\partial\theta}
-
\frac{\partial H}{\partial\theta}
\]
gives the corresponding Leslie sensitivity of reproductive potential
for
\[
\theta\in
\{m_1,\ldots,m_d,b_1,\ldots,b_{d-1}\}
\]
whenever the relevant fertility coefficient is positive.

This limiting statement must not be confused with a derivative in a
retention direction at the Leslie matrix itself. At the Leslie boundary
there is no positive retention coefficient to vary: making \(c_i=0\)
positive creates a new retention link. Its one-sided behaviour is
therefore described by equations~\eqref{BoundaryGrowthRetention}--
\eqref{BoundarySingularSensitivity}, not by the interior formula
\eqref{SensitivityC}.

The convergence of the entropy sensitivities in the Leslie parameter
directions is proved in Appendix~\ref{Corollary1}. The statements for \(r\)
follow from the standard dominant-eigenvalue formula, and those for
\(\Phi\) follow from \eqref{SensitivityIdentity}.

For an open-group Leslie matrix, the same reduction is made without
removing terminal retention. Consider the specialisation
\[
c_2=\cdots=c_{d-1}=0,
\qquad
c_d=\rho>0.
\]
For changes to its existing links,
\[
\mathcal R=\{d\},
\]
equations~\eqref{GrowthSensitivityFormulae} and
\eqref{SensitivityM}--\eqref{PotentialSensitivityC} apply with this
specialisation.

In particular, for the terminal open-group retention coefficient,
\begin{equation}
\frac{\partial r}{\partial\rho}
=
\frac{v_du_d}{\lambda},
\label{OpenGroupGrowthSensitivity}
\end{equation}
\begin{equation}
\frac{\partial H}{\partial\rho}
=
\frac{C_dq_d}{\lambda-\rho}
-
\frac{D_d}{\lambda}
+
Bv_du_d,
\label{OpenGroupEntropySensitivity}
\end{equation}
and
\begin{equation}
\frac{\partial\Phi}{\partial\rho}
=
-\frac{C_dq_d}{\lambda-\rho}
+
\frac{D_d}{\lambda}
+
\left(
\frac1\lambda-B
\right)v_du_d.
\label{OpenGroupPotentialSensitivity}
\end{equation}
The sensitivities with respect to positive fertility coefficients and
to progression coefficients are obtained from the general formulas by
the same specialisation.

By contrast, changing one of the internal coefficients
\[
c_i=0,
\qquad
2\leq i<d,
\]
to a positive value creates a new retention link; it is not a change
to an existing open-group link. Its one-sided behaviour is
described by equations~\eqref{BoundaryGrowthRetention}--
\eqref{BoundarySingularSensitivity}.

The open-group case and its precise relation to the general
existing-link formulas are stated in
Appendix~\ref{app:OpenGroupSensitivity}.

The formulas above provide, for every admissible demographic
parameter, the corresponding response of population growth,
evolutionary entropy and reproductive potential. They therefore give
the quantities required for the empirical comparisons developed in
the next section.

\subsection{A global scaling check}
\label{subsec:uniform-scaling}

The coefficientwise formulas also satisfy a useful whole-matrix check.
Multiply every positive demographic coefficient by the same factor,
while leaving every zero coefficient equal to zero. This rescaling
preserves the demographic support. The relative transition
probabilities, and hence \(H\), do not change, whereas growth and
reproductive potential each acquire the same change in logarithmic
scale. To express this property locally, let \(\mathcal P_+\) denote
the indexed collection of positive demographic parameters of \(L\),
\[
\mathcal P_{+}
=
\{m_j:m_j>0\}
\cup
\{b_j:1\leq j\leq d-1\}
\cup
\{c_i:c_i>0\}.
\]
Then
\begin{equation}
\sum_{\theta\in\mathcal P_{+}}
\theta\frac{\partial r}{\partial\theta}
=
1,
\qquad
\sum_{\theta\in\mathcal P_{+}}
\theta\frac{\partial H}{\partial\theta}
=
0,
\qquad
\sum_{\theta\in\mathcal P_{+}}
\theta\frac{\partial\Phi}{\partial\theta}
=
1.
\label{UniformScalingIdentities}
\end{equation}
These three identities describe a biological scaling property and
also provide exact checks on any numerical implementation. Their derivation is
given in Appendix~\ref{Proof:UniformScaling}.
\endgroup

\section{Examples}\label{sec:5}

The results developed in the preceding sections allow the Malthusian
parameter \(r\), evolutionary entropy \(H\), reproductive potential
\(\Phi\), their decompositions by biological process, and their
demographic sensitivities to be calculated from the same matrix. The
examples below illustrate three complementary
features of the theory: the effect of a downstream post-reproductive
stage, the distinct coefficientwise responses of \(r\), \(H\), and
\(\Phi\) in a general Lefkovitch matrix, and the open-group Leslie
specialisation.

All three empirical matrices considered below have a one-year
projection interval. The corresponding values of \(r\), \(H\), and
\(\Phi\) are therefore expressed on the same temporal scale. The
matrices were obtained from the COMADRE Animal Matrix Database
\cite{SalgueroGomez2016COMADRE}.

Since
\[
r=\log\lambda,
\]
one has, for every demographic parameter \(\theta\),
\[
\frac{\partial\lambda}{\partial\theta}
=
\lambda\frac{\partial r}{\partial\theta}.
\]
Thus the sensitivity networks of \(\lambda\) and \(r\) differ only by
the common positive factor \(\lambda\). Moreover,
\[
\frac{\theta}{\lambda}
\frac{\partial\lambda}{\partial\theta}
=
\theta\frac{\partial r}{\partial\theta}.
\]
The usual elasticity of \(\lambda\) is therefore exactly the
parameter-scaled sensitivity of \(r\). Separate networks for
\(\lambda\) would consequently contain no additional biological
information and are not displayed.

\subsection{Killer whale (\emph{Orcinus orca})}

As a first illustration, we consider the Lefkovitch matrix
associated with the killer whale \emph{Orcinus orca}
(COMADRE matrix ID 249595), drawn from the pod-specific demographic
analysis of Brault and Caswell \cite{BraultCaswell1993}. The principal
demographic quantities are summarised in Table~\ref{tab:OrcaIdentity}.

\begin{table}[ht]
\centering
\caption{Growth, reproductive potential and evolutionary entropy for
\emph{Orcinus orca}.}
\label{tab:OrcaIdentity}
\begin{tabular}{lc}
\hline
Quantity & Value\\
\hline
Dominant eigenvalue $\lambda$ & 1.036410\\
Malthusian parameter $r$ & 0.035767\\
Reproductive potential $\Phi$ & $-0.264518$\\
Evolutionary entropy $H$ & 0.300285\\
\hline
\end{tabular}
\end{table}

Separating transition and retention as in
equations~\eqref{EntropyDecomposition}--\eqref{RetentionEntropy},
introduced in Subsection~\ref{subsec:decomposition}, gives the values
shown in Table~\ref{tab:OrcaEntropy}.

\begin{table}[ht]
\centering
\caption{Evolutionary-entropy decomposition for \emph{Orcinus orca}.}
\label{tab:OrcaEntropy}
\begin{tabular}{lccc}
\hline
Component & Value & Fraction & Percent (\%)\\
\hline
Transition entropy $H_{\rm tr}$ & 0.007821 & 0.026045 & 2.604\\
Retention entropy $H_{\rm ret}$ & 0.292464 & 0.973956 & 97.396\\
\hline
Total entropy $H$ & 0.300285 & 1.000000 & 100.000\\
\hline
\end{tabular}
\end{table}

Stage retention accounts for approximately \(97\%\) of the total
evolutionary entropy, whereas less than \(3\%\) is associated with the
transition component. This calculation identifies which biological
process contributes to \(H\), but does not by itself determine the local
sensitivity of \(H\) to individual demographic coefficients.

The terminal stage is post-reproductive and forms a downstream
component of the demographic dynamics. For changes confined to this
post-reproductive component that keep it downstream of the
reproductive stages, the corresponding sensitivities of \(r\) and \(H\)
vanish. This property is visible in
Figures~\ref{fig:OrcaGrowthNetwork} and
\ref{fig:OrcaEntropyNetwork}. Since
\[
d\Phi=dr-dH,
\]
the reproductive-potential sensitivity also vanishes for the same
changes.

This statement concerns changes to links already present in the
post-reproductive part of the life cycle. Creating a fertility link
from a post-reproductive stage is a different change, described by the
small-\(\varepsilon\) formulas~\eqref{BoundaryGrowthFertility}--
\eqref{BoundarySingularSensitivity} introduced in
Subsection~\ref{subsec:new-links}.

\begin{figure}[ht]
\centering
\includegraphics[width=0.80\textwidth]{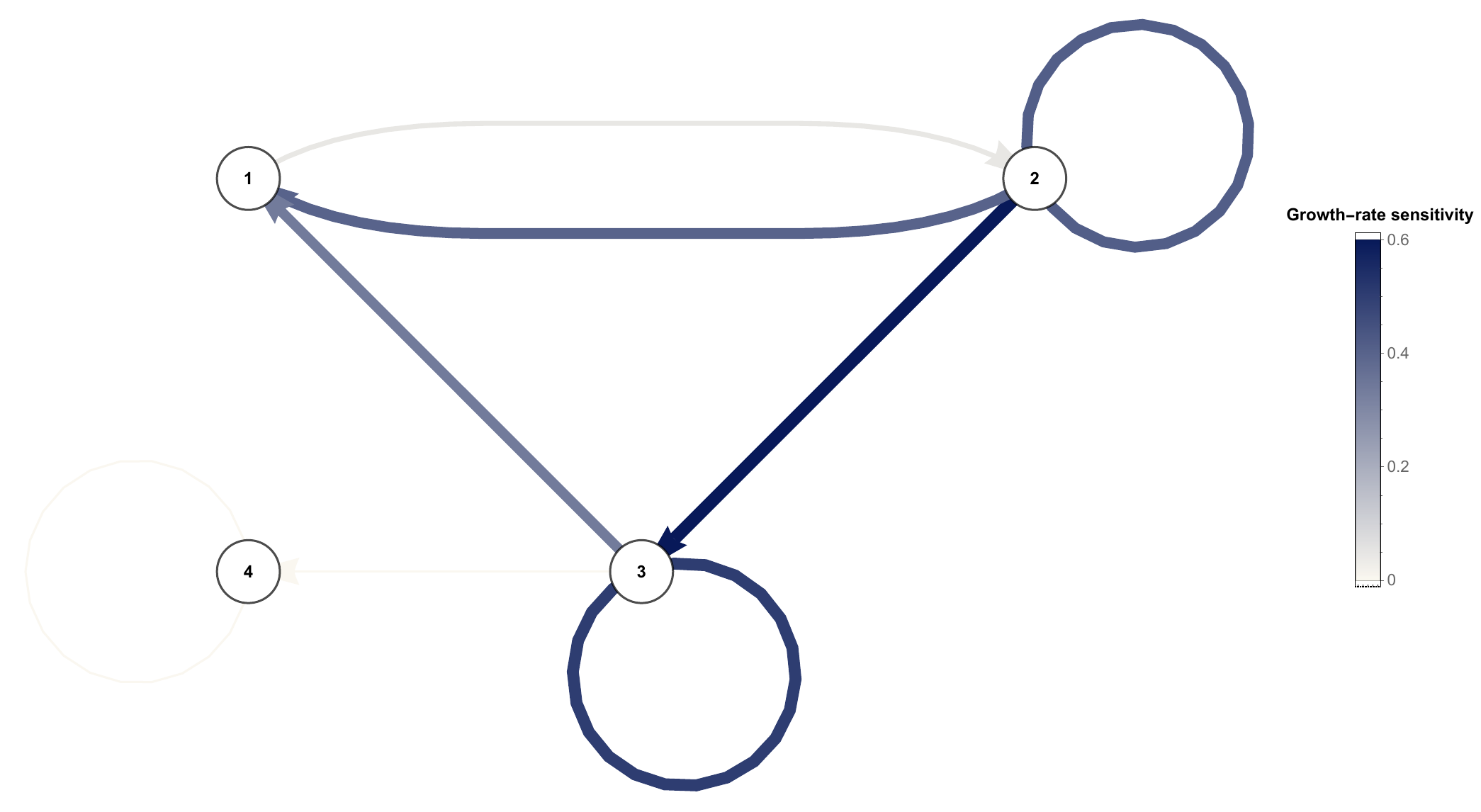}
\caption{Sensitivity network for population growth in
\emph{Orcinus orca}. Edge thickness is proportional to the magnitude
of the sensitivity coefficient. Changes confined to the
terminal post-reproductive component have zero growth sensitivity
when the downstream block structure is preserved.}
\label{fig:OrcaGrowthNetwork}
\end{figure}

\begin{figure}[ht]
\centering
\includegraphics[width=0.80\textwidth]{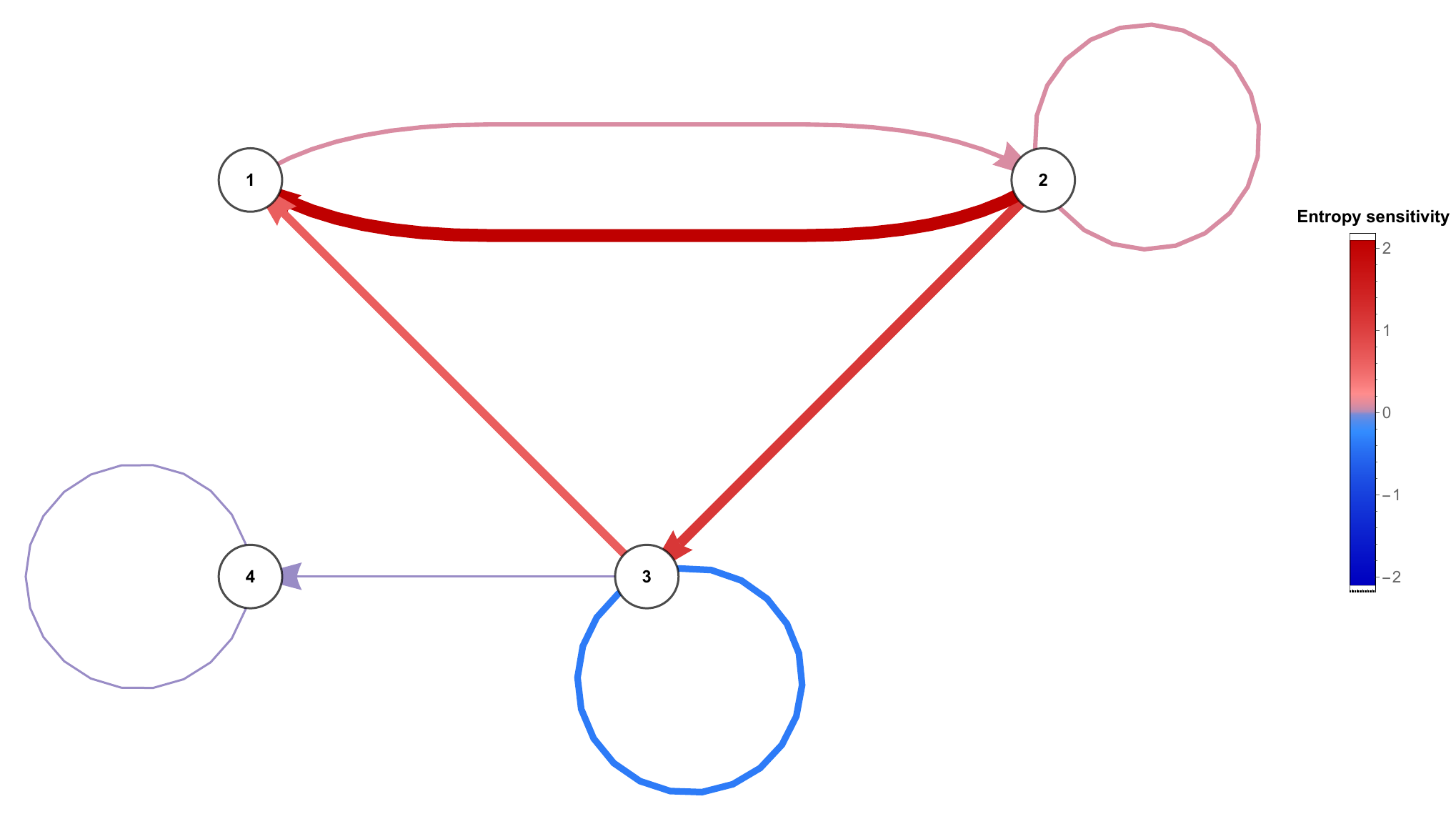}
\caption{Sensitivity network for evolutionary entropy in
\emph{Orcinus orca}. Edge thickness is proportional to the magnitude
of the sensitivity coefficient. The terminal post-reproductive
component has zero entropy sensitivity when its downstream position is
preserved, despite the strong overall contribution of retention
to \(H\).}
\label{fig:OrcaEntropyNetwork}
\end{figure}

\FloatBarrier

\subsection{Flathead catfish (\emph{Pylodictis olivaris})}

We next consider the Lefkovitch matrix associated with the flathead
catfish \emph{Pylodictis olivaris} (COMADRE matrix ID 248200),
corresponding to the Ocmulgee River population model of Sakaris and
Irwin \cite{SakarisIrwin2010}. The principal demographic quantities
are summarised in Table~\ref{tab:PylodictisIdentity}.

\begin{table}[ht]
\centering
\caption{Growth, reproductive potential and evolutionary entropy for
\emph{Pylodictis olivaris}.}
\label{tab:PylodictisIdentity}
\begin{tabular}{lc}
\hline
Quantity & Value\\
\hline
Dominant eigenvalue $\lambda$ & 0.999353\\
Malthusian parameter $r$ & $-0.000648$\\
Reproductive potential $\Phi$ & $-0.616745$\\
Evolutionary entropy $H$ & 0.616098\\
\hline
\end{tabular}
\end{table}

The entropy decomposition is shown in
Table~\ref{tab:PylodictisEntropy}.

\begin{table}[ht]
\centering
\caption{Evolutionary-entropy decomposition for
\emph{Pylodictis olivaris}.}
\label{tab:PylodictisEntropy}
\begin{tabular}{lccc}
\hline
Component & Value & Fraction & Percent (\%)\\
\hline
Transition entropy $H_{\rm tr}$ & 0.151030 & 0.245140 & 24.514\\
Retention entropy $H_{\rm ret}$ & 0.465067 & 0.754860 & 75.486\\
\hline
Total entropy $H$ & 0.616098 & 1.000000 & 100.000\\
\hline
\end{tabular}
\end{table}

Approximately three quarters of the evolutionary entropy are
associated with stage retention, while approximately one quarter is
associated with the transition component. The catfish therefore
provides a less extreme example than the killer whale, while still
showing a strong contribution from stage persistence.

To compare demographic coefficients whose numerical magnitudes differ
substantially, we use the parameter-scaled sensitivities
\[
e_\theta^{(X)}
=
\theta\frac{\partial X}{\partial\theta},
\qquad
X\in\{r,H,\Phi\}.
\]
For \(X=r\), this quantity is exactly the conventional elasticity of
the dominant eigenvalue,
\[
e_\theta^{(r)}
=
\frac{\theta}{\lambda}
\frac{\partial\lambda}{\partial\theta}.
\]
For \(H\) and \(\Phi\), the quantities
\(e_\theta^{(H)}\) and \(e_\theta^{(\Phi)}\) are parameter-scaled
sensitivities; no division by \(H\) or \(\Phi\) is made.

Their algebraic sums are fixed by the uniform-scaling identities
\eqref{UniformScalingIdentities}, explained in
Subsection~\ref{subsec:uniform-scaling}:
\[
\sum_{\theta\in\mathcal P_{+}}e_\theta^{(r)}=1,
\qquad
\sum_{\theta\in\mathcal P_{+}}e_\theta^{(H)}=0,
\qquad
\sum_{\theta\in\mathcal P_{+}}e_\theta^{(\Phi)}=1.
\]

To summarise the overall magnitude of the coefficientwise response,
define
\[
E_X^{\rm abs}
=
\sum_{\theta\in\mathcal P_+}
\left|
e_\theta^{(X)}
\right|,
\]
where \(\mathcal P_+\) denotes the set of positive demographic
coefficients. These are sums of absolute responses and must therefore
not be confused with algebraic sums.

\begin{table}[ht]
\centering
\caption{Aggregate absolute parameter-scaled sensitivities for
\emph{Pylodictis olivaris}.}
\label{tab:PylodictisElasticities}
\begin{tabular}{lc}
\hline
Quantity & Value\\
\hline
Growth, $E_r^{\rm abs}$ & 1.000000\\
Evolutionary entropy, $E_H^{\rm abs}$ & 0.729447\\
Reproductive potential, $E_\Phi^{\rm abs}$ & 1.069890\\
\hline
\end{tabular}
\end{table}

Figures~\ref{fig:PylodictisGrowthElasticity},
\ref{fig:PylodictisEntropyElasticity}, and
\ref{fig:PylodictisPotentialElasticity}
display the corresponding coefficientwise response networks for
\(r\), \(H\), and \(\Phi\). The three quantities are evaluated on the
same demographic matrix but need not assign the same relative
importance, or even the same sign, to a given demographic
coefficient.

\begin{figure}[ht]
\centering
\includegraphics[width=0.85\textwidth]{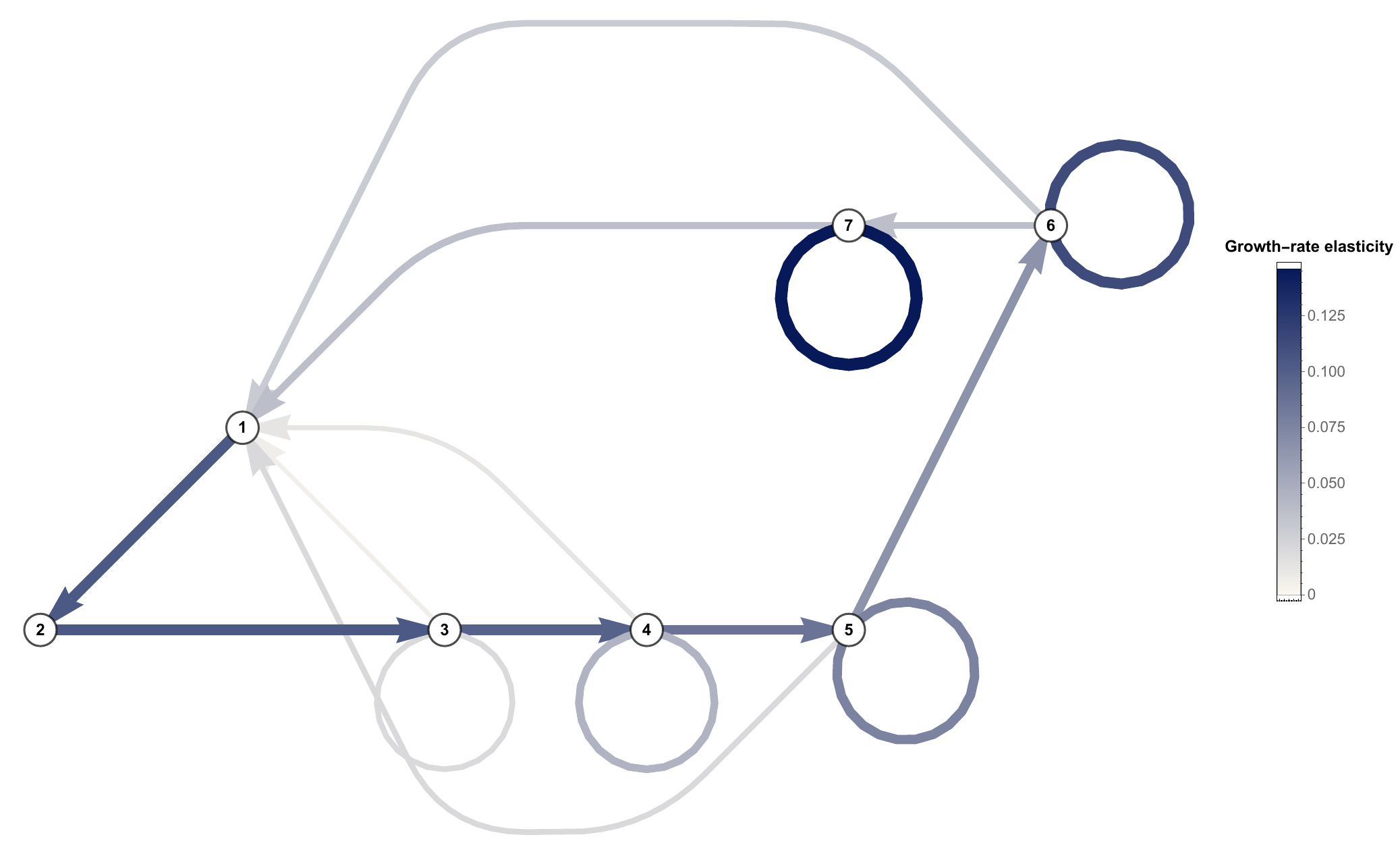}
\caption{Parameter-scaled sensitivity network for population growth in
\emph{Pylodictis olivaris}. Edge thickness is proportional to
\(\left|e_\theta^{(r)}\right|\). Since
\(e_\theta^{(r)}=(\theta/\lambda)
(\partial\lambda/\partial\theta)\), these coefficients are also the
usual elasticities of the dominant eigenvalue.}
\label{fig:PylodictisGrowthElasticity}
\end{figure}

\begin{figure}[ht]
\centering
\includegraphics[width=0.85\textwidth]{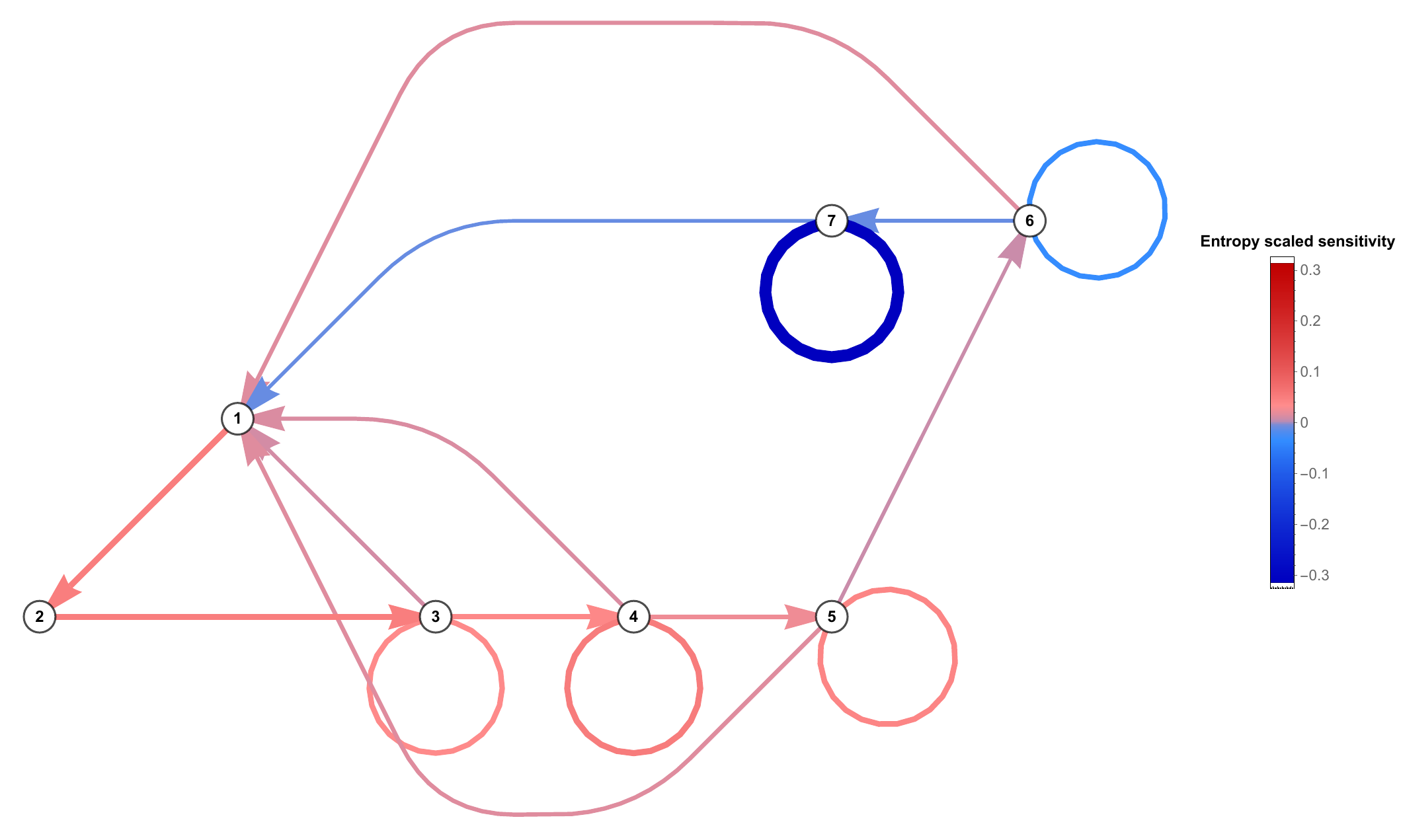}
\caption{Parameter-scaled sensitivity network for evolutionary
entropy in \emph{Pylodictis olivaris}. Edge thickness is proportional
to \(\left|e_\theta^{(H)}\right|\). Positive coefficients correspond
to changes that increase entropy, whereas negative coefficients
correspond to changes that decrease entropy.}
\label{fig:PylodictisEntropyElasticity}
\end{figure}

The growth and entropy networks display markedly different
coefficientwise patterns. The demographic coefficients having the
largest effect on \(r\) need not be those having the largest effect
on \(H\), and the two responses need not have the same sign.

The corresponding response of reproductive potential is fixed
parameter by parameter by
\[
e_\theta^{(\Phi)}
=
e_\theta^{(r)}
-
e_\theta^{(H)}.
\]
The potential network therefore displays the complementary demographic
response required by the identity
\[
r=H+\Phi.
\]
In particular, a demographic coefficient can have weak net influence
on \(r\) while inducing larger and opposing responses in \(H\) and
\(\Phi\).

\begin{figure}[ht]
\centering
\includegraphics[width=0.85\textwidth]{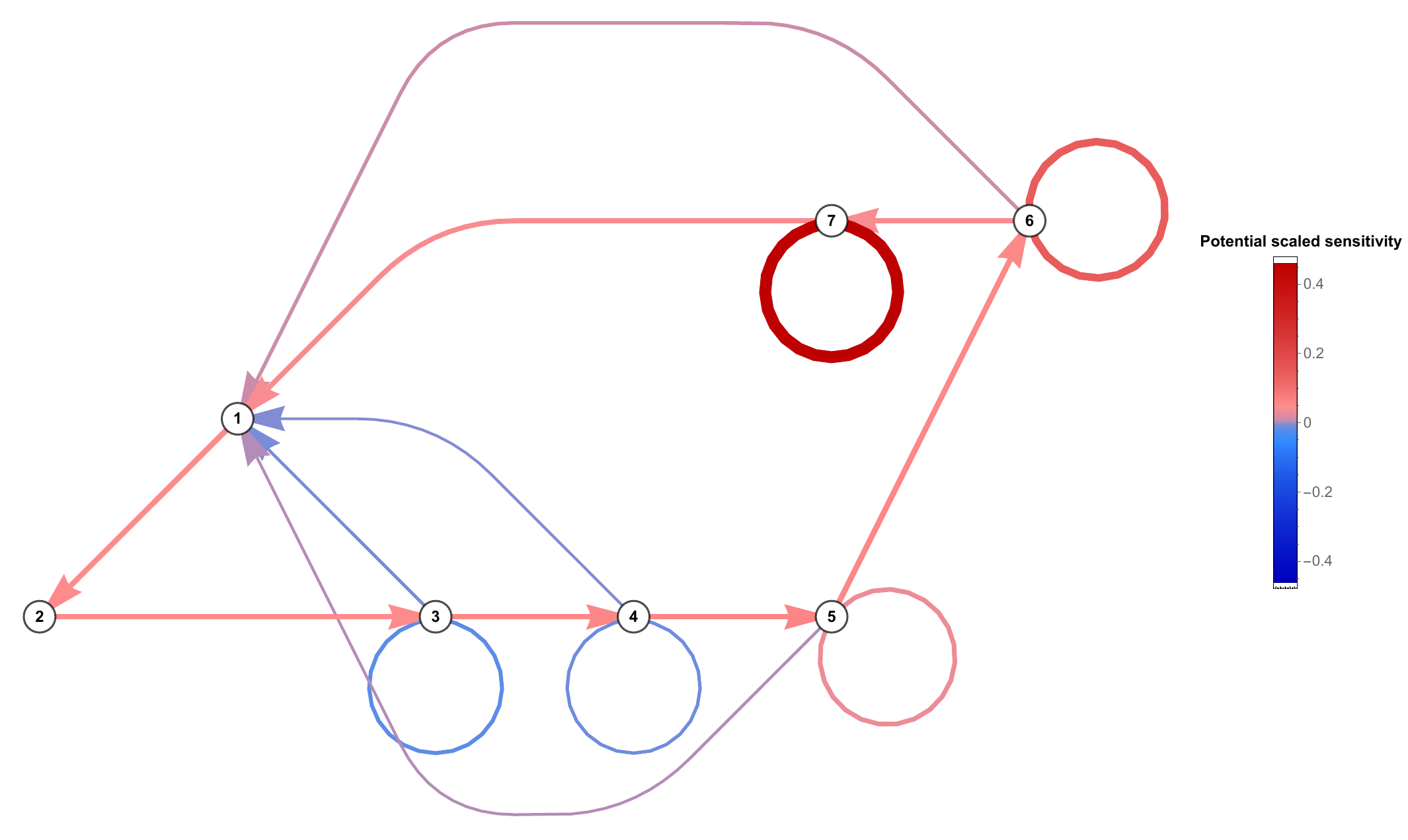}
\caption{Parameter-scaled sensitivity network for reproductive
potential in \emph{Pylodictis olivaris}. Edge thickness is
proportional to \(\left|e_\theta^{(\Phi)}\right|\). Coefficientwise,
the response satisfies
\(e_\theta^{(\Phi)}=e_\theta^{(r)}-e_\theta^{(H)}\).}
\label{fig:PylodictisPotentialElasticity}
\end{figure}

The aggregate absolute responses reinforce this distinction:
\[
E_r^{\rm abs}=1.000000,
\qquad
E_H^{\rm abs}=0.729447,
\qquad
E_\Phi^{\rm abs}=1.069890.
\]
These quantities measure the overall magnitude of the local
coefficientwise response and not its algebraic direction.

\FloatBarrier

\subsection{Giraffe (\emph{Giraffa camelopardalis})}

As a final illustration, we consider the open-group Leslie matrix
of the giraffe \emph{Giraffa camelopardalis} (COMADRE matrix
ID 240422), associated in COMADRE with the demographic analysis of
Strauss et al.~\cite{StraussEtAl2015}. For this 1975--1977 Seronera
matrix, the COMADRE metadata indicate that the parameter values were
based on Pellew~\cite{Pellew1983c}. The matrix has a one-year
projection interval and is
\[
L=
\begin{pmatrix}
0 & 0 & 0.30\\
0.42 & 0 & 0\\
0 & 0.92 & 0.95
\end{pmatrix}.
\]
It has the open-group Leslie structure
\[
m_1=m_2=0,
\qquad
m_3=0.30,
\]
\[
b_1=0.42,
\qquad
b_2=0.92,
\]
and
\[
c_2=0,
\qquad
c_3=\rho=0.95.
\]
Thus reproduction occurs exclusively from the terminal open class,
and retention occurs only in that same class. This matrix therefore
provides a particularly transparent application of the open-group
formulas~\eqref{OpenGroupS}--\eqref{OpenGroupPhi} from
Subsection~\ref{subsec:open-group-representation} and of the terminal
retention sensitivities~\eqref{OpenGroupGrowthSensitivity}--
\eqref{OpenGroupPotentialSensitivity} from
Subsection~\ref{subsec:sensitivity-reductions}.

The principal demographic quantities are summarised in
Table~\ref{tab:GiraffeIdentity}.

\begin{table}[ht]
\centering
\caption{Growth, reproductive potential and evolutionary entropy for
\emph{Giraffa camelopardalis}.}
\label{tab:GiraffeIdentity}
\begin{tabular}{lc}
\hline
Quantity & Value\\
\hline
Dominant eigenvalue $\lambda$ & 1.054289\\
Malthusian parameter $r$ & 0.052867\\
Reproductive potential $\Phi$ & $-0.216536$\\
Evolutionary entropy $H$ & 0.269403\\
\hline
\end{tabular}
\end{table}

Because fertility occurs only in the terminal class, the distribution
of reproduction is concentrated entirely in stage 3:
\[
q_1=q_2=0,
\qquad
q_3=1.
\]
Consequently,
\[
H_{\rm tr}=0,
\qquad
H=H_{\rm ret}.
\]
The numerical decomposition is shown in
Table~\ref{tab:GiraffeEntropy}.

\begin{table}[ht]
\centering
\caption{Evolutionary-entropy decomposition for
\emph{Giraffa camelopardalis}.}
\label{tab:GiraffeEntropy}
\begin{tabular}{lccc}
\hline
Component & Value & Fraction & Percent (\%)\\
\hline
Transition entropy $H_{\rm tr}$ & 0.000000 & 0.000000 & 0.000\\
Retention entropy $H_{\rm ret}$ & 0.269403 & 1.000000 & 100.000\\
\hline
Total entropy $H$ & 0.269403 & 1.000000 & 100.000\\
\hline
\end{tabular}
\end{table}

The reason for the zero transition component is simple. The
distribution of reproduction contains
only one positive entry and therefore contributes no demographic
diversity to \(H_{\rm tr}\). All evolutionary entropy in this example
is consequently contained in the retention component.

This statement concerns the sources that make up the value of \(H\);
it does not determine the sign of the derivative of \(H\) with respect
to retention. Decomposition and sensitivity answer different
biological questions.

\begin{figure}[ht]
\centering
\includegraphics[width=0.80\textwidth]{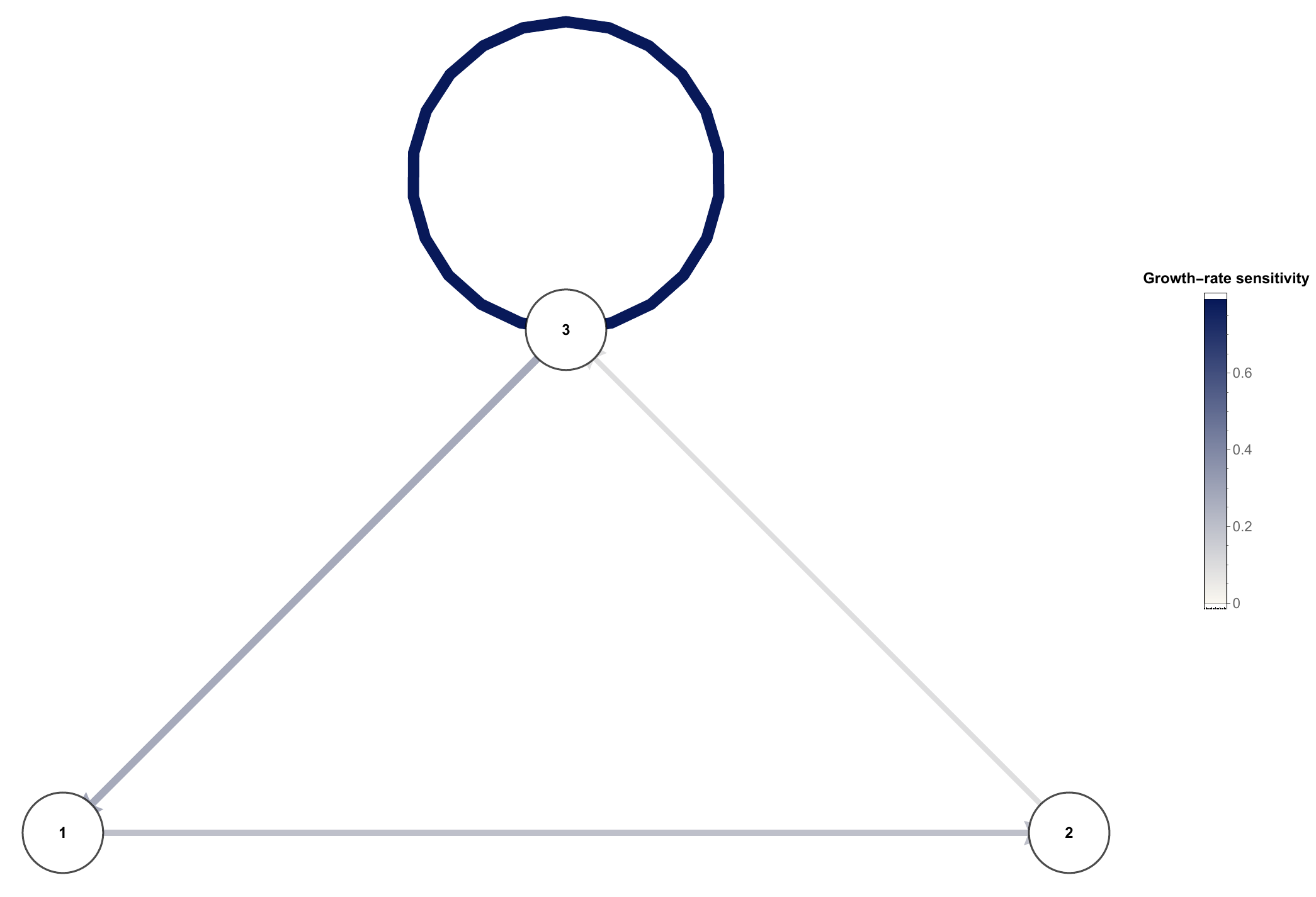}
\caption{Sensitivity network for population growth in
\emph{Giraffa camelopardalis}. Edge thickness is proportional to the
magnitude of the sensitivity coefficient. The terminal retention
coefficient has positive growth sensitivity.}
\label{fig:GiraffeGrowthNetwork}
\end{figure}

For the terminal retention coefficient
\[
\rho=c_3=0.95,
\]
the explicit open-group formulas give
\[
\frac{\partial r}{\partial\rho}
=
0.791849>0,
\]
whereas
\[
\frac{\partial H}{\partial\rho}
=
-0.327709<0.
\]
The corresponding reproductive-potential sensitivity is
\[
\frac{\partial\Phi}{\partial\rho}
=
1.11956>0.
\]
Numerically,
\[
0.791849
\simeq
-0.327709+1.11956,
\]
in agreement with
\[
\frac{\partial r}{\partial\rho}
=
\frac{\partial H}{\partial\rho}
+
\frac{\partial\Phi}{\partial\rho}.
\]

Thus increasing terminal retention increases the Malthusian growth
rate while decreasing evolutionary entropy. Reproductive potential,
by contrast, increases sufficiently to more than compensate for the
decrease in \(H\), resulting in the positive net response of \(r\).

This is especially informative because the decomposition gives
\[
H=H_{\rm ret}.
\]
Retention therefore accounts for the entire value of evolutionary
entropy in this matrix, while increasing the retention coefficient
produces a negative local response of \(H\). The example demonstrates
particularly clearly that decomposition and sensitivity answer
different demographic questions.

\begin{figure}[ht]
\centering
\includegraphics[width=0.80\textwidth]{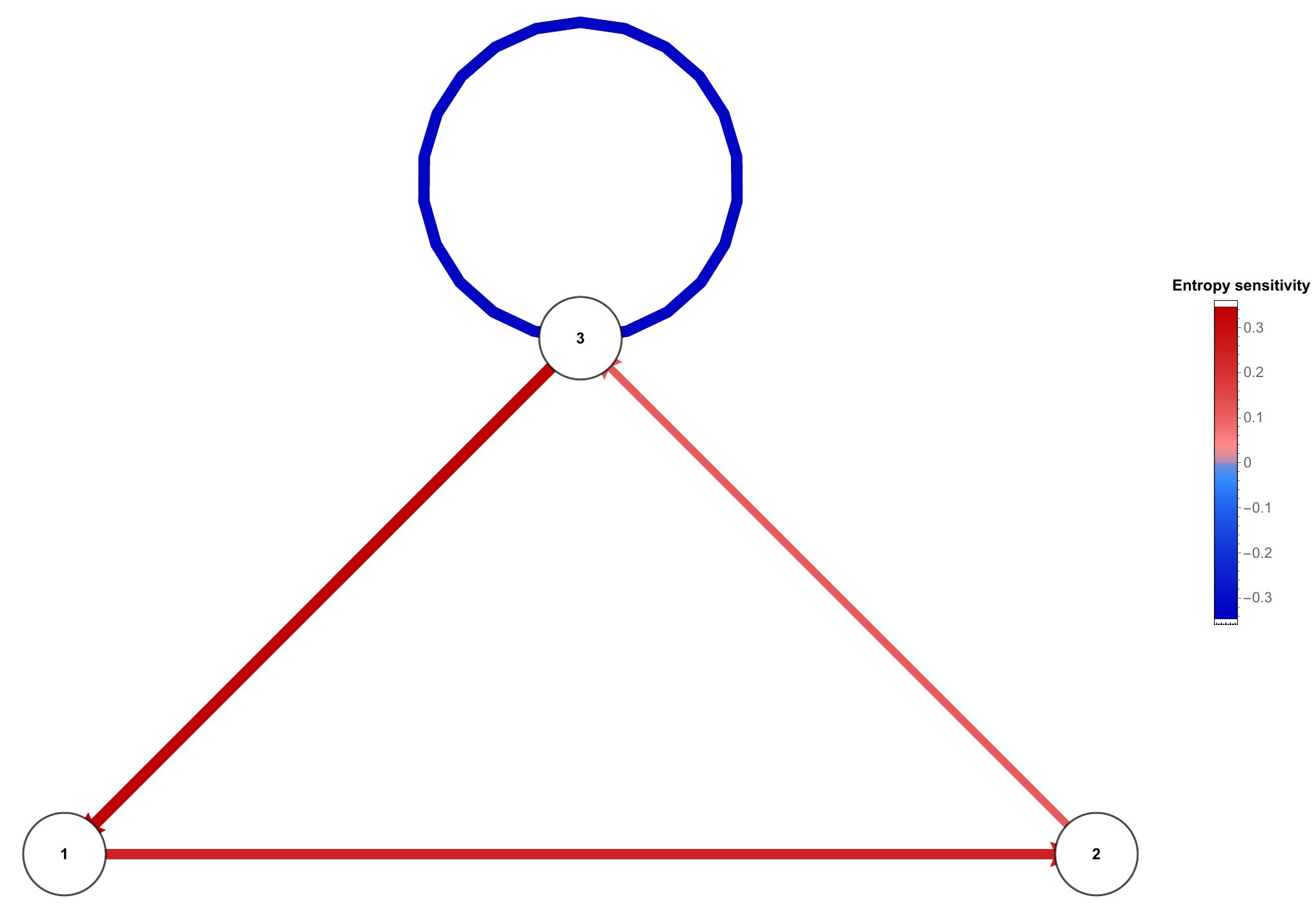}
\caption{Sensitivity network for evolutionary entropy in
\emph{Giraffa camelopardalis}. Edge thickness is proportional to the
magnitude of the sensitivity coefficient. The negative sensitivity of
the terminal retention coefficient contrasts with its positive effect
on population growth.}
\label{fig:GiraffeEntropyNetwork}
\end{figure}

\FloatBarrier

\begingroup
\color{black}

\section{Conclusions}\label{sec:6}

This paper develops a joint sensitivity theory for the Malthusian
parameter \(r\), evolutionary entropy \(H\), and reproductive
potential \(\Phi\) in Lefkovitch matrices. Starting from the dominant
eigenvalue and the associated Markov chain, we obtained explicit
formulas for the long-run stage distribution, generation time,
evolutionary entropy, and reproductive potential. The fundamental
relation
\[
r=H+\Phi
\]
therefore connects three quantities that can be represented and
analysed within the same demographic dynamics.

Stage retention contributes terms that are absent
from the classical Leslie model. Both evolutionary entropy and
reproductive potential admit exact decompositions,
\[
H=H_{\mathrm{tr}}+H_{\mathrm{ret}},
\qquad
\Phi=\Phi_{\mathrm{tr}}+\Phi_{\mathrm{ret}},
\]
with distinct interpretations: for \(H\), the decomposition separates
sources of demographic uncertainty, whereas for \(\Phi\) it partitions
the demographic links into fertility--progression and retention
contributions. This grouping does not make the components independent:
they remain coupled through the dominant eigenvalue, the distribution
of reproduction across stages, and the long-run stage probabilities. As
retention tends to zero, the additional retention contributions
disappear and the classical Leslie expressions are recovered.

For changes to existing demographic links, we derived explicit
sensitivity formulas for fertility, progression, and retention. The
sensitivity of \(r\) is given by the standard formula for a simple
dominant eigenvalue. The response of \(H\) contains a direct part from
changes in reproductive, progression, and retention probabilities,
and an indirect part caused by the change in the dominant eigenvalue.
The corresponding sensitivity of
reproductive potential follows exactly from
\[
d\Phi=dr-dH.
\]
This provides a common sensitivity description of population growth,
evolutionary entropy, and reproductive potential without requiring a
separate differential theory for \(\Phi\).

A different behaviour occurs when a fertility or retention coefficient
leaves zero and creates a new link. The growth rate still has a finite
one-sided derivative, but the new link
produces logarithmically singular changes in \(H\) and \(\Phi\):
\[
H(\varepsilon)-H(0)
=
-\kappa\varepsilon\log\varepsilon+O(\varepsilon),
\qquad
\Phi(\varepsilon)-\Phi(0)
=
\kappa\varepsilon\log\varepsilon+O(\varepsilon),
\]
where \(\kappa>0\) denotes the appropriate coefficient
\(\kappa_k^{(m)}\) or \(\kappa_k^{(c)}\).
Thus continuity at a zero coefficient does not imply a finite
sensitivity there. Any sensitivity analysis of evolutionary entropy
must therefore distinguish between changing an existing demographic
link and creating a new one.

The biological examples illustrate that decomposition and sensitivity
describe different properties of a demographic system. Retention may
account for most, or even all, of the observed evolutionary entropy
without implying that increasing retention increases \(H\); similarly,
the demographic coefficients having the largest effects on population
growth need not have the largest effects on evolutionary entropy or
reproductive potential. The post-reproductive structure of the killer
whale and the open terminal class of the giraffe further illustrate
how the general theory accommodates biologically important departures
from the classical Leslie setting.

The resulting theory places Lefkovitch, Leslie, and open-group Leslie
models within a common treatment of \(r\), \(H\), and \(\Phi\).
Because Lefkovitch and open-group Leslie structures occur frequently
in empirical demographic databases, the formulas obtained here provide
a direct basis for comparative analyses of how fertility, progression,
and stage persistence contribute to population growth and to its
entropy--potential decomposition.

\endgroup

\section*{Data availability}

The demographic matrices used in the illustrative examples were
obtained from the public COMADRE Animal Matrix Database. The
structural classification reported in Table~\ref{tab:databaseprevalence}
was obtained from the public COMADRE Animal Matrix Database and
COMPADRE Plant Matrix Database, accessed on 29 May 2026. All
theoretical results are derived analytically in the article, and the
data required to reproduce the empirical classification and numerical
examples are available from these public databases.

\section*{Code availability}

The Wolfram Mathematica notebook implementing the analytical formulas
developed in this article and reproducing the numerical examples,
sensitivity calculations, finite-difference checks, boundary-link tests,
and figures is publicly available on Zenodo at
\url{https://doi.org/10.5281/zenodo.21902324}.

\section*{Funding}

This work was supported by Fundação para a Ciência e a Tecnologia (FCT), Portugal, through the Centro de Análise Matemática, Geometria e Sistemas Dinâmicos (CAMGSD), under the projects:

\begin{itemize}
\item UID/04459/2025 (DOI: 10.54499/UID/04459/2025)
\item UID/PRR/04459/2025 (DOI: 10.54499/UID/PRR/04459/2025),
\end{itemize}
financed by national funds and by the European Union through the Recovery and Resilience Plan (PRR) – NextGenerationEU.

\section*{Competing Interests}

The author declares that he has no competing interests.

\bigskip

\section*{Author Contributions}
Single authorship.

\bigskip

\appendix

\begingroup
\color{black}

\section{Formal results and mathematical details}
\label{app:SensitivityProofs}

This appendix provides the formal mathematical layer corresponding to
the narrative development in Sections~\ref{sec:2}--\ref{sec:4}.
Each theorem, proposition, corollary, or lemma below identifies the
precise equations and the main-text subsection in which the result is
introduced and interpreted. Conversely, every substantive derivation
in the main text points to the relevant appendix subsection.

\subsection{Table of symbols and notation}
\label{app:notation}

\begin{table}[ht]
\centering
\small
\begin{tabular}{ll}
\toprule
Symbol & Description\\
\midrule

$d$ & Number of demographic stages \\

$L$ & Lefkovitch projection matrix \\

$m_j$ & Fertility coefficient of stage $j$ \\

$b_j$ & Progression coefficient from stage $j$ to stage $j+1$ \\

$c_j$ & Retention coefficient of stage $j$ \\

$\lambda$ & Dominant eigenvalue \\

$r$ & Malthusian parameter, $r=\log\lambda$ \\

$u$ & Right eigenvector associated with $\lambda$ \\

$v$ & Left eigenvector associated with $\lambda$, normalised by $vu=1$ \\

$U$ & $\operatorname{diag}(u_1,\ldots,u_d)$ \\

$P$ & Transition matrix of the associated Markov chain \\

$R_j$ & Relative weight of stage $j$ before fertility is included \\

$q_j$ & Probability that reproduction occurs in stage $j$ \\

$Q_j$ & Tail probability, $\sum_{k=j}^{d}q_k$ \\

$\eta_j$ & Progression probability in the associated Markov chain \\

$\rho_j$ & Retention probability in the associated Markov chain \\

$\rho$ & Terminal retention coefficient, $\rho=c_d$, in an open-group Leslie matrix \\

$\pi_j$ & Stationary probability of stage $j$ \\

$T$ & Generation time \\

$S$ & Numerator of evolutionary entropy, $H=S/T$ \\

$H$ & Evolutionary entropy \\

$\Phi$ & Reproductive potential \\

$H_{\mathrm{tr}}$ & Reproductive-transition component of evolutionary entropy \\

$H_{\mathrm{ret}}$ & Stage-persistence component of evolutionary entropy \\

$\Phi_{\mathrm{tr}}$ & Transition component of reproductive potential \\

$\Phi_{\mathrm{ret}}$ & Retention component of reproductive potential \\

$h_j$ & Local retention entropy term \\

$\mathcal F$ & Set of stages with positive fertility, $\{j:m_j>0\}$ \\

$\mathcal R$ & Set of stages with positive retention, $\{j:c_j>0\}$ \\

$\theta$ & Generic demographic parameter \\

$\mathcal P_+$ & Indexed collection of positive demographic parameters \\

$A_j,C_j,D_j,B$ & Auxiliary quantities in the sensitivity formulas \\

$e_{\theta}^{(X)}$ & Parameter-scaled sensitivity, $\theta\,\partial X/\partial\theta$ \\

$E_X^{\mathrm{abs}}$ & Aggregate absolute parameter-scaled sensitivity \\

\bottomrule
\end{tabular}
\caption{Main notation used throughout the paper.}
\label{tab:notation1}
\end{table}
\FloatBarrier



\subsection{The Markov chain associated with the projection matrix}
\label{app:PerronMarkovResults}

\subsubsection{Generalised Euler--Lotka formula}
\label{Proof:EulerLotka}

This result justifies the construction in
Subsection~\ref{subsec:reproductive-probabilities}, specifically
equations~\eqref{RjDefinition}--\eqref{GeneralizedEulerLotka}.

\begin{theorem}[Generalised Euler--Lotka formula]
\label{thm:EulerLotka}
Let \(L\) be an irreducible Lefkovitch matrix of the form
\eqref{eq:Lefkovitch}, with dominant eigenvalue \(\lambda\) and
positive right eigenvector \(u\). Define \(R_j\) and \(q_j\) by
\eqref{RjDefinition} and \eqref{qjDefinition}. Then the components of
\(u\) satisfy \eqref{PerronRepresentation}, the weights \(q_j\) form
the probability distribution in \eqref{qjNormalization}, and their
normalisation is equivalently the generalised Euler--Lotka equation
\eqref{GeneralizedEulerLotka}.
\end{theorem}

\begin{proof}

Starting from the dominant-eigenvalue equation for the matrix introduced
in Section~\ref{sec:2},
\[
Lu=\lambda u
\]
gives, for \(j=2,\ldots,d\),
\begin{equation}
b_{j-1}u_{j-1}+c_ju_j=\lambda u_j.
\label{EigenRecurrence}
\end{equation}
Hence
\begin{equation}
u_j
=
u_{j-1}\frac{b_{j-1}}{\lambda-c_j}.
\label{EigenStep}
\end{equation}
Iteration yields
\[
u_j=u_1R_j,
\qquad
j=1,\ldots,d.
\]
Substitution into the first row of the eigenvalue equation gives
\[
\sum_{j=1}^{d}\frac{m_jR_j}{\lambda}=1.
\]
Since \(q_j=m_jR_j/\lambda\), this is precisely
\[
\sum_{j=1}^{d}q_j=1.
\]
Thus equations~\eqref{PerronRepresentation}--
\eqref{GeneralizedEulerLotka} in
Subsection~\ref{subsec:reproductive-probabilities} follow.

\end{proof}


\subsubsection{Stationary distribution and generation time}
\label{proof:StationaryDistribution}

The following statement establishes the long-run stage probabilities
described in Subsection~\ref{subsec:stationary-pathway}, using the
progression and retention probabilities in
\eqref{EtaRhoDefinition}, the tail probabilities in
\eqref{TailDistribution}, and the claimed stationary distribution and
generation time in \eqref{StationaryDistribution}--\eqref{TDefinition}.

\begin{theorem}[Stationary distribution]
\label{thm:StationaryDistribution}
Let \(P\) be the Markov transition matrix defined in
\eqref{MarkovMatrix}, and define \(\eta_i\), \(\rho_i\), and \(Q_i\)
by \eqref{EtaRhoDefinition} and \eqref{TailDistribution}. Then \(P\)
has the stationary distribution given in
\eqref{StationaryDistribution}, where the normalising constant \(T\)
is given by \eqref{TDefinition}.
\end{theorem}

\begin{proof}

Using the reproductive probabilities defined in
Subsection~\ref{subsec:reproductive-probabilities}, for
\(i=2,\ldots,d\) consider the partition
\[
\{1,\ldots,i-1\}
\qquad\text{and}\qquad
\{i,\ldots,d\}.
\]
The stationary flux from the first set to the second is
\(\pi_1Q_i\), whereas the reverse flux is \(\pi_i\eta_i\).
Hence
\begin{equation}
\pi_i\eta_i=\pi_1Q_i.
\label{BalanceRelation}
\end{equation}
Therefore
\begin{equation}
\pi_i=\frac{\pi_1Q_i}{\eta_i}.
\label{PiFormula}
\end{equation}
Normalization of the stationary distribution gives
\[
1
=
\pi_1
\left(
1+\sum_{i=2}^{d}\frac{Q_i}{\eta_i}
\right)
=
\pi_1T.
\]
Thus
\[
\pi_1=\frac1T,
\qquad
\pi_i=\frac{Q_i}{\eta_iT},
\qquad
i=2,\ldots,d.
\]
This is precisely equations~\eqref{StationaryDistribution}--
\eqref{TDefinition} in Subsection~\ref{subsec:stationary-pathway}.

\end{proof}


\subsubsection{Explicit formula for evolutionary entropy}
\label{Proof:EntropyFormula}

This theorem proves
equations~\eqref{EntropyRepresentation}--\eqref{EntropyS} in
Subsection~\ref{subsec:entropy-potential}. It uses the stationary
distribution and generation time established in
Theorem~\ref{thm:StationaryDistribution}.

\begin{theorem}[Explicit formula for evolutionary entropy]
\label{thm:EntropyFormula}
For the irreducible Lefkovitch matrix \(L\), evolutionary entropy is
given by \(H=S/T\) in \eqref{EntropyRepresentation}, where \(S\) is
the expression in \eqref{EntropyS} and \(T\)
is the generation time in \eqref{TDefinition}.
\end{theorem}

\begin{proof}

For the associated Markov chain defined in
Subsection~\ref{subsec:stationary-pathway}, evolutionary entropy is,
by definition,
\begin{equation}
H
=
-\sum_{i,j}\pi_iP_{ij}\log P_{ij}.
\label{EntropyDefinition}
\end{equation}
The first row of \(P\) is \((q_1,\ldots,q_d)\), whereas for
\(i\geq2\) its only possible non-zero entries are
\(\eta_i\) and \(\rho_i\). Hence
\[
H
=
-\pi_1\sum_{j=1}^{d}q_j\log q_j
-
\sum_{i=2}^{d}
\pi_i
\left(
\eta_i\log\eta_i+\rho_i\log\rho_i
\right),
\]
with zero terms interpreted by continuity. Substituting
\[
\pi_1=\frac1T,
\qquad
\pi_i=\frac{Q_i}{\eta_iT},
\]
gives
\[
H
=
\frac1T
\left[
-\sum_{j=1}^{d}q_j\log q_j
-
\sum_{i=2}^{d}
\frac{Q_i}{\eta_i}
\left(
\eta_i\log\eta_i+\rho_i\log\rho_i
\right)
\right],
\]
which is \eqref{EntropyRepresentation}--\eqref{EntropyS}.

\end{proof}


\subsubsection{Reproductive potential and the fundamental identity}
\label{app:PotentialFormula}

This result establishes the definition
\eqref{PotentialDefinition}, the explicit formulas
\eqref{PotentialRepresentation}--\eqref{PotentialExplicit}, and the
identity \eqref{GrowthEntropyPotentialIdentity}, all introduced in
Subsection~\ref{subsec:entropy-potential}.

\begin{corollary}[Explicit formula for reproductive potential]
\label{cor:PotentialFormula}
For the irreducible Lefkovitch matrix \(L\), the reproductive
potential defined by \eqref{PotentialDefinition} is given by
\eqref{PotentialRepresentation} and, directly in terms of the
demographic coefficients, by \eqref{PotentialExplicit}. Together with
the entropy in Theorem~\ref{thm:EntropyFormula}, it satisfies
\(r=H+\Phi\), as stated in
\eqref{GrowthEntropyPotentialIdentity}.
\end{corollary}

\begin{proof}
The positive entries of the associated Markov matrix are
\(P_{1j}=q_j\), \(P_{i,i-1}=\eta_i\), and
\(P_{ii}=\rho_i\), with the index restrictions stated immediately
after \eqref{PotentialDefinition}. Substitution in that definition
gives \eqref{PotentialRepresentation}; substitution of
\eqref{StationaryDistribution} gives \eqref{PotentialExplicit}.
Finally,
\[
\log L_{ij}-\log P_{ij}
=
\log\lambda+\log u_i-\log u_j.
\]
Averaging with respect to \(\pi_iP_{ij}\) cancels the two eigenvector
terms by stationarity and yields
\eqref{GrowthEntropyPotentialIdentity}.
\end{proof}


\subsubsection{Open-group Leslie formulas}
\label{app:OpenGroupRepresentation}

The following corollary establishes the special case developed in
Subsection~\ref{subsec:open-group-representation}.

\begin{corollary}[Open-group Leslie matrices]
\label{cor:opengroup}
If \(c_2=\cdots=c_{d-1}=0\) and \(c_d=\rho>0\), the entropy
functional, generation time, evolutionary entropy, and reproductive
potential are respectively given by
\eqref{OpenGroupS}--\eqref{OpenGroupPhi}. The identity
\(\log\lambda=H+\Phi\) remains valid.
\end{corollary}

\begin{proof}
Under this specialisation, \(\eta_i=1\) and \(\rho_i=0\) for
\(2\leq i<d\), while
\(\eta_d=(\lambda-\rho)/\lambda\) and
\(\rho_d=\rho/\lambda\). Substitution in
\eqref{EntropyS}, \eqref{TDefinition}, and
\eqref{PotentialExplicit} gives
\eqref{OpenGroupS}--\eqref{OpenGroupPhi}. The final identity follows
from \eqref{GrowthEntropyPotentialIdentity}.
\end{proof}


\subsection{Separating transitions from stage retention and recovering the Leslie model}
\label{app:DecompositionResults}

\subsubsection{Evolutionary-entropy decomposition}
\label{app:EntropyDecomposition}

This theorem establishes the separation introduced in
Subsection~\ref{subsec:decomposition}, specifically
equations~\eqref{EntropyDecomposition}--\eqref{RetentionEntropy}.

\begin{theorem}[Entropy decomposition]
\label{thm:EntropyDecomposition}
The evolutionary entropy of an irreducible Lefkovitch matrix admits
the exact decomposition \eqref{EntropyDecomposition}, with transition
and retention components given by \eqref{ReproductiveEntropy} and
\eqref{RetentionEntropy}.
\end{theorem}

\begin{proof}
In the Markov-chain entropy formula \eqref{EntropyDefinition}, the
first row of \(P\) contributes \(H_{\mathrm{tr}}\), while rows
\(i=2,\ldots,d\) contribute \(H_{\mathrm{ret}}\). Their sum is the
explicit entropy formula proved in
Theorem~\ref{thm:EntropyFormula}, which gives
\eqref{EntropyDecomposition}--\eqref{RetentionEntropy}.
\end{proof}


\subsubsection{Reproductive-potential decomposition}
\label{app:PotentialDecomposition}

This proposition establishes
equations~\eqref{PotentialDecomposition}--\eqref{RetentionPotential}
in Subsection~\ref{subsec:decomposition}.

\begin{proposition}[Reproductive-potential decomposition]
\label{prop:PotentialDecomposition}
The reproductive potential admits the exact decomposition
\eqref{PotentialDecomposition}, with non-retention and retention
components given by \eqref{TransitionPotential} and
\eqref{RetentionPotential}.
\end{proposition}

\begin{proof}
Separate in \eqref{PotentialRepresentation} the fertility and
progression entries from the entries for remaining in the same stage.
The first group is
exactly \eqref{TransitionPotential}, and the second is
\eqref{RetentionPotential}; their sum gives
\eqref{PotentialDecomposition}.
\end{proof}


\subsubsection{Leslie limit}
\label{Proof:LeslieLimit}

The following corollary establishes the limiting argument in
Subsection~\ref{subsec:leslie-limit}, in particular
equations~\eqref{LeslieEntropyFormula}--
\eqref{LesliePotentialFormula}.

\begin{corollary}[Leslie limit]
\label{cor:LeslieLimit}
As \(c_i\to0\) for \(i=2,\ldots,d\), the retention components
\(H_{\mathrm{ret}}\) and \(\Phi_{\mathrm{ret}}\) vanish, while \(H\)
and \(\Phi\) converge to the Leslie expressions
\eqref{LeslieEntropyFormula} and \eqref{LesliePotentialFormula}.
Consequently, \(r_L=H_L+\Phi_L=\log\lambda_L\).
\end{corollary}

\begin{proof}

Let \(c=(c_2,\ldots,c_d)\), and let \(L(c)\) converge to the
corresponding Leslie matrix as \(c_i\to0\), \(i=2,\ldots,d\).
Let \(\lambda_L\) denote the dominant eigenvalue of the limiting
Leslie matrix. By continuity of a simple dominant eigenvalue,
\[
\lambda(c)\longrightarrow\lambda_L.
\]
Consequently, for each \(j\),
\[
R_j(c)
\longrightarrow
R_j^{(0)}
=
\frac{b_1\cdots b_{j-1}}{\lambda_L^{\,j-1}},
\qquad
q_j(c)\longrightarrow
q_j^{(0)}
=
\frac{m_jb_1\cdots b_{j-1}}{\lambda_L^{\,j}},
\]
and therefore
\[
Q_i(c)\longrightarrow Q_i^{(0)},
\qquad
T(c)\longrightarrow
T_L
=
1+\sum_{i=2}^{d}Q_i^{(0)}.
\]

\begingroup
\color{black}
Moreover,
\[
\rho_i(c)\longrightarrow0,
\qquad
\eta_i(c)\longrightarrow1.
\]
Hence
\[
\rho_i(c)\log\rho_i(c)\longrightarrow0,
\qquad
\eta_i(c)\log\eta_i(c)\longrightarrow0,
\]
by \(x\log x\to0\) as \(x\downarrow0\) and continuity at \(x=1\),
respectively. Since \(Q_i(c)/\eta_i(c)\) remains bounded, every
retention term in \eqref{EntropyS} vanishes in the limit.
\endgroup

Thus
\[
H(c)
\longrightarrow
-\frac1{T_L}
\sum_{j=1}^{d}
q_j^{(0)}\log q_j^{(0)}
=
H_L.
\]

In \eqref{PotentialExplicit}, the terms involving retention vanish
because \(\rho_i\log c_i\to0\), whereas \(q_j\), \(Q_i\), and \(T\)
converge to their Leslie counterparts. This gives
\eqref{LesliePotentialFormula}. The limiting form of
\eqref{GrowthEntropyPotentialIdentity} then yields
\(r_L=H_L+\Phi_L=\log\lambda_L\), completing the proof of the result
stated in Subsection~\ref{subsec:leslie-limit}.

\end{proof}


\subsection{Sensitivities for existing demographic links}
\label{app:SupportPreservingResults}

Throughout this subsection, an \emph{existing-link perturbation} has
the meaning given in Subsection~\ref{subsec:support-preserving}: the
positive fertility and retention coefficients may vary, but zero
coefficients remain zero. Thus the index sets \(\mathcal F\) and
\(\mathcal R\) do not change. Changes that create a new link are
treated separately in Appendix~\ref{app:BoundaryResults}.

\subsubsection{Sensitivity of the Malthusian parameter}
\label{app:GrowthSensitivity}

The following theorem establishes
equations~\eqref{GrowthDifferential}--
\eqref{GrowthSensitivityFormulae}, introduced and interpreted in
Subsection~\ref{subsec:sensitivity-start}.

\begin{theorem}[Sensitivity of the Malthusian parameter]
\label{thm:GrowthSensitivity}
Let \(u\) and \(v\) be right and left eigenvectors associated with the
dominant eigenvalue of the
irreducible Lefkovitch matrix \(L\), normalised by \(vu=1\). For an
infinitesimal perturbation \(dL\), the dominant eigenvalue and
Malthusian parameter satisfy \eqref{GrowthDifferential}. In
particular, the fertility, progression, and retention sensitivities
are those in \eqref{GrowthSensitivityFormulae}, with the index ranges
specified immediately after that equation.
\end{theorem}

\begin{proof}
Since the dominant eigenvalue is simple, differentiating
\(Lu=\lambda u\) and multiplying on the left by \(v\) gives
\(d\lambda=v(dL)u\). Division by \(\lambda\), using
\(r=\log\lambda\), yields \eqref{GrowthDifferential}. Applying this
identity separately to entries \(L_{1k}=m_k\),
\(L_{k+1,k}=b_k\), and \(L_{kk}=c_k\) gives
\eqref{GrowthSensitivityFormulae}.
\end{proof}


\subsubsection{Entropy response in reproductive and retention probabilities}
\label{Lemma1}

This lemma establishes the first step of the calculation in
Subsubsection~\ref{subsubsec:normalized-pathways}, namely
equations~\eqref{EntropyDifferential}--\eqref{DiDefinition}.

\begin{lemma}[Entropy response in probability variables]
\label{lem:DifferentialRepresentation}
Along an existing-link perturbation, the entropy differential is
given by \eqref{EntropyDifferential}, with the coefficients \(C_j\)
and \(D_i\) defined by \eqref{CjDefinition} and
\eqref{DiDefinition}.
\end{lemma}

\begin{proof}

Using the entropy functional \eqref{EntropyFunctional} from
Subsubsection~\ref{subsubsec:normalized-pathways}, along a
existing-link perturbation,
\(q_j=0\) for \(j\notin\mathcal F\) and
\(h_i=0\) for \(i\notin\mathcal R\), so
\[
S
=
-\sum_{j\in\mathcal F}q_j\log q_j
+
\sum_{i\in\mathcal R}\frac{Q_i}{\eta_i}h_i.
\]
Since \(H=S/T\),
\begin{equation}
dH=\frac1T(dS-H\,dT).
\label{QuotientRule}
\end{equation}
Furthermore,
\begin{equation}
dT
=
\sum_{i=2}^{d}\frac{dQ_i}{\eta_i}
-
\sum_{i\in\mathcal R}
\frac{Q_i}{\eta_i^2}\,d\eta_i,
\label{dTFormula}
\end{equation}
with
\begin{equation}
dQ_i
=
\sum_{\substack{j\in\mathcal F\\j\geq i}}dq_j.
\label{dQiFormula}
\end{equation}
Differentiating \(S\) gives
\begin{align}
dS
&=
-\sum_{j\in\mathcal F}(1+\log q_j)\,dq_j
\nonumber\\
&\quad+
\sum_{i\in\mathcal R}
\left(
\frac{dQ_i}{\eta_i}
-
\frac{Q_i}{\eta_i^2}\,d\eta_i
\right)h_i
+
\sum_{i\in\mathcal R}
\frac{Q_i}{\eta_i}\,dh_i .
\label{dSFormula}
\end{align}
For \(i\in\mathcal R\), since
\(\rho_i=1-\eta_i\),
\begin{equation}
dh_i
=
\log\!\left(\frac{\rho_i}{\eta_i}\right)d\eta_i.
\label{dhiFormula}
\end{equation}

Substitution of
\eqref{dTFormula}--\eqref{dhiFormula}
into \eqref{QuotientRule} shows that the coefficient of
\(dq_j\), \(j\in\mathcal F\), is
\[
\frac1T
\left[
-(1+\log q_j)
+
\sum_{i=2}^{j}
\frac{h_i-H}{\eta_i}
\right]
=
C_j.
\]
The coefficient of \(d\eta_i\), \(i\in\mathcal R\), is
\[
\frac{Q_i}{T\eta_i^2}
\left[
H-h_i+
\eta_i\log\!\left(\frac{\rho_i}{\eta_i}\right)
\right].
\]
Using
\[
h_i=-\eta_i\log\eta_i-\rho_i\log\rho_i,
\qquad
\eta_i+\rho_i=1,
\]
the bracket reduces to \(H+\log\rho_i\). Hence the coefficient is
\(D_i\), proving \eqref{EntropyDifferential}.

\end{proof}


\subsubsection{Parameter differentials}
\label{Lemma2}

This lemma establishes the second step of the calculation in
Subsubsection~\ref{subsubsec:parameter-differentials}, specifically
equations~\eqref{AjDefinition}--\eqref{detaFormula}.

\begin{lemma}[Parameter differentials]
\label{lem:ParameterDifferentials}
For an existing-link perturbation, let \(A_j\) be defined by
\eqref{AjDefinition}. Then the reproductive weights and
progression probabilities satisfy \eqref{dqjFormula} and
\eqref{detaFormula}.
\end{lemma}

\begin{proof}

Starting from the reproductive weight \eqref{qjDefinition} and the
stage weight \eqref{RjDefinition}, for \(j\in\mathcal F\),
\[
q_j
=
\frac{m_j}{\lambda}
\prod_{h=1}^{j-1}
\frac{b_h}{\lambda-c_{h+1}}.
\]
Taking logarithmic differentials gives
\[
\frac{dq_j}{q_j}
=
\frac{dm_j}{m_j}
+
\sum_{h=1}^{j-1}\frac{db_h}{b_h}
+
\sum_{h=1}^{j-1}
\frac{dc_{h+1}}{\lambda-c_{h+1}}
-
\left(
\frac1\lambda+
\sum_{h=1}^{j-1}
\frac1{\lambda-c_{h+1}}
\right)d\lambda.
\]
By the definition of \(A_j\), this is exactly
\eqref{dqjFormula}. Under an existing-link perturbation,
\(dc_i=0\) whenever \(i\notin\mathcal R\).

Finally,
\[
\eta_i=1-\frac{c_i}{\lambda},
\]
so, for \(i\in\mathcal R\),
\[
d\eta_i
=
\frac{c_i}{\lambda^2}\,d\lambda
-
\frac1\lambda\,dc_i,
\]
which is \eqref{detaFormula}.

\end{proof}


\subsubsection{Entropy differential in demographic parameters}
\label{Theorem1}

This theorem corresponds to the assembled demographic differential
\eqref{IntermediateDifferential} in
Subsubsection~\ref{subsubsec:entropy-response}, with \(B\) defined by
\eqref{BDefinition}.

\begin{theorem}[Entropy differential]
\label{thm:EntropyDifferential}
For every existing-link perturbation, define \(B\) by
\eqref{BDefinition}. Then the entropy response in the original
fertility, progression, and retention parameters is exactly
\eqref{IntermediateDifferential}.
\end{theorem}

\begin{proof}

Substitute the parameter differentials
\eqref{dqjFormula}--\eqref{detaFormula}, established in
Lemma~\ref{lem:ParameterDifferentials}, into the probability-based entropy
differential \eqref{EntropyDifferential}, established in
Lemma~\ref{lem:DifferentialRepresentation}. Collecting respectively
the fertility, progression, and retention terms gives
\[
\sum_{j\in\mathcal F}
C_jq_j\frac{dm_j}{m_j},
\]
\[
\sum_{k=1}^{d-1}
\left(
\sum_{\substack{j\in\mathcal F\\j\geq k+1}}
C_jq_j
\right)
\frac{db_k}{b_k},
\]
and
\[
\sum_{k\in\mathcal R}
\left[
\frac1{\lambda-c_k}
\sum_{\substack{j\in\mathcal F\\j\geq k}}
C_jq_j
-
\frac{D_k}{\lambda}
\right]dc_k.
\]
The remaining coefficient of \(d\lambda\) is
\[
-\sum_{j\in\mathcal F}C_jq_jA_j
+
\sum_{i\in\mathcal R}
D_i\frac{c_i}{\lambda^2}
=
B.
\]
This proves \eqref{IntermediateDifferential}, the formula stated and
interpreted in Subsubsection~\ref{subsubsec:entropy-response}.

\end{proof}


\subsubsection{Explicit entropy sensitivities}
\label{Theorem2}

This theorem establishes the coefficientwise formulas
\eqref{TotalSensitivityFormula}--\eqref{SensitivityC} in
Subsubsection~\ref{subsubsec:coefficientwise-sensitivities}.

\begin{theorem}[Explicit entropy-sensitivity formulas]
\label{thm:ExplicitSensitivity}
For every existing-link perturbation, the total entropy response
is \eqref{TotalSensitivityFormula}, and its sensitivities with respect
to positive fertility coefficients, progression coefficients, and
positive retention coefficients are respectively
\eqref{SensitivityM}, \eqref{SensitivityB}, and
\eqref{SensitivityC}.
\end{theorem}

\begin{proof}

By Theorem~\ref{thm:GrowthSensitivity}, since the dominant eigenvalue is
simple,
\begin{equation}
d\lambda=v(dL)u.
\label{PerronFormula}
\end{equation}
For an existing-link perturbation,
\[
d\lambda
=
\sum_{k\in\mathcal F}v_1u_k\,dm_k
+
\sum_{k=1}^{d-1}v_{k+1}u_k\,db_k
+
\sum_{k\in\mathcal R}v_ku_k\,dc_k.
\]
Substituting this expression into
\eqref{IntermediateDifferential} and collecting the coefficients
of \(dm_k\), \(db_k\), and \(dc_k\) gives respectively
\eqref{SensitivityM}, \eqref{SensitivityB}, and
\eqref{SensitivityC}.

\end{proof}


\subsubsection{Explicit reproductive-potential sensitivities}
\label{app:PotentialSensitivity}

The following corollary corresponds to
equations~\eqref{PotentialSensitivityM}--
\eqref{PotentialSensitivityC} in
Subsubsection~\ref{subsubsec:coefficientwise-sensitivities}.

\begin{corollary}[Explicit reproductive-potential sensitivity formulas]
\label{cor:PotentialSensitivity}
For every existing-link perturbation, the sensitivities of
reproductive potential with respect to positive fertility
coefficients, progression coefficients, and positive retention
coefficients are respectively \eqref{PotentialSensitivityM},
\eqref{PotentialSensitivityB}, and \eqref{PotentialSensitivityC}.
\end{corollary}

\begin{proof}
Differentiate the identity \eqref{GrowthEntropyPotentialIdentity} to
obtain \eqref{SensitivityIdentity}. Subtract the entropy sensitivities
\eqref{SensitivityM}--\eqref{SensitivityC} from the growth
sensitivities \eqref{GrowthSensitivityFormulae}, parameter by
parameter. This gives
\eqref{PotentialSensitivityM}--\eqref{PotentialSensitivityC}.
\end{proof}


\subsection{Creation of new demographic links}
\label{app:BoundaryResults}

\subsubsection{Behaviour when a new link is created}
\label{Proof:BoundaryLinks}

This proposition establishes the small-\(\varepsilon\) formulas for
new fertility and retention links in
\eqref{BoundaryGrowthFertility}--
\eqref{BoundarySingularSensitivity}, introduced in
Subsection~\ref{subsec:new-links}.

\begin{proposition}[Creation of a new link]
\label{prop:BoundaryLinks}
Let \(L_0\) be an irreducible Lefkovitch matrix with the notation used
in Subsection~\ref{subsec:new-links}. If a zero fertility coefficient
\(m_k\) is replaced by \(\varepsilon>0\), then
\eqref{BoundaryGrowthFertility}--
\eqref{BoundaryFertilityCoefficient} hold. If a zero retention
coefficient \(c_k\) is replaced by \(\varepsilon>0\), then
\eqref{BoundaryGrowthRetention}--
\eqref{BoundaryRetentionCoefficient} hold. In both cases the
one-sided behaviour is that of
\eqref{BoundarySingularSensitivity}.
\end{proposition}

\begin{proof}

Because the dominant eigenvalue of \(L_0\) is simple, the dominant
eigenvalue \(\lambda(\varepsilon)\), the Malthusian parameter
\(r(\varepsilon)\), and the associated normalised eigenvectors depend
smoothly on \(\varepsilon\).

Suppose first that \(m_k=0\) is replaced by \(\varepsilon>0\).
Since all other coefficients are fixed,
\[
q_k(\varepsilon)
=
\frac{R_k^{(0)}}{\lambda_0}\,\varepsilon
+
O(\varepsilon^2).
\]
The only logarithmic entropy contribution produced by the new link is
therefore
\[
-\pi_1q_k\log q_k
=
-\frac{R_k^{(0)}}{\lambda_0T_0}
\,\varepsilon\log\varepsilon
+
O(\varepsilon),
\]
which gives \eqref{BoundaryEntropyFertility} with
\[
\kappa_k^{(m)}
=
\frac{R_k^{(0)}}{\lambda_0T_0}.
\]
Since \(r=H+\Phi\) and \(r\) is differentiable,
the singular term in \(\Phi\) has the opposite sign, giving
\eqref{BoundaryPotentialFertility}. The expansion of \(r\) follows directly
from \eqref{GrowthSensitivityFormulae}.

Now suppose that \(c_k=0\) is replaced by \(\varepsilon>0\). Then
\[
\rho_k(\varepsilon)
=
\frac{\varepsilon}{\lambda_0}
+
O(\varepsilon^2),
\]
and the only logarithmic contribution is
\[
-\pi_k
\left(
\eta_k\log\eta_k+\rho_k\log\rho_k
\right)
=
-\frac{\pi_k^{(0)}}{\lambda_0}
\,\varepsilon\log\varepsilon
+
O(\varepsilon).
\]
Hence
\[
\kappa_k^{(c)}
=
\frac{\pi_k^{(0)}}{\lambda_0}
=
\frac{Q_k^{(0)}}{\lambda_0T_0},
\]
which gives \eqref{BoundaryEntropyRetention}. Again,
\(r=H+\Phi\) gives the opposite logarithmic contribution to
\(\Phi\), and the finite expansion of \(r\) follows from
\eqref{GrowthSensitivityFormulae}.

Dividing the expansions by \(\varepsilon\) and letting
\(\varepsilon\downarrow0\), the relation
\[
-\log\varepsilon\longrightarrow+\infty
\]
gives the stated one-sided derivative behaviour.

\end{proof}


\subsection{Leslie and open-group cases and the scaling check}
\label{app:ReductionResults}

\subsubsection{Reduction to Leslie sensitivities}
\label{Corollary1}

This corollary establishes the convergence described in
Subsection~\ref{subsec:sensitivity-reductions}. Its proof identifies
the precise limits of the quantities entering
\eqref{SensitivityM}--\eqref{SensitivityC}.

\begin{corollary}[Leslie sensitivity reduction]
\label{Cor:LeslieCase}
As \(c_i\to0\), \(i=2,\ldots,d\), while the same fertility links remain
present and each \(c_i\) approaches zero through positive values, the
sensitivities with respect to fertility and progression
converge to the corresponding Leslie sensitivities. The entropy
sensitivities reduce to the Demetrius--Gundlach--Ziehe formulas, the
growth sensitivities reduce to the classical Leslie formulas, and
the potential sensitivities follow from
\eqref{SensitivityIdentity}. Creating a positive retention coefficient
from a Leslie matrix is governed by
Proposition~\ref{prop:BoundaryLinks}.
\end{corollary}

\begin{proof}

\begingroup
\color{black}
Let \(c=(c_2,\ldots,c_d)\), and let \(L(c)\) converge to the
corresponding Leslie matrix as \(c_i\to0\), \(i=2,\ldots,d\), with
the same positive fertility coefficients \(\mathcal F\) throughout
the limit.
\endgroup
By continuity of the simple dominant eigenvalue,
\[
\lambda(c)\to\lambda_L.
\]
From the limits established in Appendix~\ref{Proof:LeslieLimit},
\[
q_j(c)\to q_j^{(0)},
\qquad
Q_i(c)\to Q_i^{(0)},
\qquad
T(c)\to T_L,
\qquad
H(c)\to H_L,
\]
while
\[
\eta_i(c)\to1,
\qquad
\rho_i(c)\to0,
\qquad
h_i(c)\to0.
\]
Consequently, for \(j\in\mathcal F\),
\[
C_j(c)
\longrightarrow
C_j^{(0)}
=
\frac1{T_L}
\left[
-(1+\log q_j^{(0)})-(j-1)H_L
\right],
\qquad
A_j(c)\longrightarrow\frac{j}{\lambda_L}.
\]

\begingroup
\color{black}
It remains only to control the retention contribution to \(B\).
Since
\[
D_i(c)c_i
=
\frac{Q_i(c)c_i}
{T(c)\eta_i(c)^2}
\left(
H(c)+\log\rho_i(c)
\right)
\]
and
\[
c_i=\lambda(c)\rho_i(c),
\]
we obtain
\[
D_i(c)c_i
=
\frac{Q_i(c)\lambda(c)}
{T(c)\eta_i(c)^2}
\left[
\rho_i(c)H(c)
+
\rho_i(c)\log\rho_i(c)
\right].
\]
The prefactor remains bounded, while
\[
\rho_i(c)H(c)\longrightarrow0,
\qquad
\rho_i(c)\log\rho_i(c)\longrightarrow0.
\]
Hence
\[
D_i(c)c_i\longrightarrow0.
\]
\endgroup
Consequently
\[
B(c)\longrightarrow
B^{(0)}
=
-\sum_{j\in\mathcal F}
C_j^{(0)}q_j^{(0)}
\frac{j}{\lambda_L}.
\]

Substitution of these limits into
\eqref{SensitivityM} and \eqref{SensitivityB} gives the corresponding
Leslie fertility and progression sensitivities. Retention coefficients
are absent from the Leslie model; creating a positive retention
coefficient from \(c_i=0\) is instead described by
equations~\eqref{BoundaryGrowthRetention}--
\eqref{BoundarySingularSensitivity}.

\end{proof}


\subsubsection{Open-group Leslie sensitivities}
\label{app:OpenGroupSensitivity}

This corollary establishes the open-group sensitivity case in
Subsection~\ref{subsec:sensitivity-reductions}, in particular
equations~\eqref{OpenGroupGrowthSensitivity}--
\eqref{OpenGroupPotentialSensitivity}.

\begin{corollary}[Open-group Leslie sensitivity]
\label{cor:OpenGroupSensitivity}
If \(c_2=\cdots=c_{d-1}=0\) and \(c_d=\rho>0\), then, for
existing-link perturbations, \(\mathcal R=\{d\}\) and the general
formulas \eqref{GrowthSensitivityFormulae},
\eqref{SensitivityM}--\eqref{SensitivityC}, and
\eqref{PotentialSensitivityM}--\eqref{PotentialSensitivityC} apply
with that specialisation. In particular, the terminal-retention
sensitivities are \eqref{OpenGroupGrowthSensitivity}--
\eqref{OpenGroupPotentialSensitivity}. Creating an internal retention
link is instead governed by Proposition~\ref{prop:BoundaryLinks}.
\end{corollary}

\begin{proof}
For an open-group Leslie matrix the only positive retention
coefficient is \(c_d=\rho\), so \(\mathcal R=\{d\}\). Substitution of
\(k=d\) and \(c_d=\rho\) in
\eqref{GrowthSensitivityFormulae}, \eqref{SensitivityC}, and
\eqref{PotentialSensitivityC} gives respectively
\eqref{OpenGroupGrowthSensitivity},
\eqref{OpenGroupEntropySensitivity}, and
\eqref{OpenGroupPotentialSensitivity}. An internal coefficient
\(c_i=0\), \(i<d\), is not an existing link, so its creation falls
under Proposition~\ref{prop:BoundaryLinks}.
\end{proof}


\subsubsection{Uniform-scaling identities}
\label{Proof:UniformScaling}

The following corollary establishes the scaling argument and the three
identities in
Subsection~\ref{subsec:uniform-scaling}, specifically
\eqref{UniformScalingIdentities}.

\begin{corollary}[Uniform-scaling identities]
\label{cor:UniformScaling}
Let \(\mathcal P_+\) be the indexed collection of positive demographic
parameters defined immediately before
\eqref{UniformScalingIdentities}. Then the parameter-scaled
sensitivities of \(r\), \(H\), and \(\Phi\) satisfy the three exact
identities in \eqref{UniformScalingIdentities}.
\end{corollary}

\begin{proof}
For \(a>0\), consider the one-parameter family
\[
L(a)=aL.
\]
Each positive demographic parameter
\(\theta\in\mathcal P_{+}\) is therefore transformed into
\[
\theta(a)=a\theta,
\]
whereas zero coefficients remain zero, so the same links remain
present.

The dominant eigenvalue satisfies
\[
\lambda(a)=a\lambda,
\]
and its left and right eigenvectors may be chosen independently of
\(a\).
Consequently,
\[
P(a)
=
\lambda(a)^{-1}U^{-1}L(a)U
=
(a\lambda)^{-1}U^{-1}(aL)U
=
P.
\]
Hence the associated Markov chain and its stationary
distribution are unchanged, and therefore
\[
H(L(a))=H(L).
\]
Moreover,
\[
r(L(a))
=
\log\lambda(a)
=
r(L)+\log a,
\]
and, using \(r=H+\Phi\),
\[
\Phi(L(a))
=
\Phi(L)+\log a.
\]

Differentiating the three identities
\(r(L(a))=r(L)+\log a\),
\(H(L(a))=H(L)\), and
\(\Phi(L(a))=\Phi(L)+\log a\)
with respect to the scaling parameter \(a\), and applying the chain
rule, gives at \(a=1\)
\[
\left.\frac{d}{da}X(L(a))\right|_{a=1}
=
\sum_{\theta\in\mathcal P_{+}}
\frac{\partial X}{\partial\theta}
\left.\frac{d\theta(a)}{da}\right|_{a=1}
=
\sum_{\theta\in\mathcal P_{+}}
\theta\frac{\partial X}{\partial\theta},
\]
for \(X\in\{r,H,\Phi\}\).
Since
\[
\left.
\frac{d}{da}r(L(a))
\right|_{a=1}=1,
\qquad
\left.
\frac{d}{da}H(L(a))
\right|_{a=1}=0,
\qquad
\left.
\frac{d}{da}\Phi(L(a))
\right|_{a=1}=1,
\]
equation~\eqref{UniformScalingIdentities} follows.
\end{proof}



\section{Computational implementation}

The formulas of Sections 2--4 were implemented in Wolfram Mathematica.
For matrices containing only a downstream post-reproductive block of the
type described in Section 2, the full matrix was retained; the zero
long-run probability of the downstream block makes its contribution to
the long-run quantities and to changes within that block vanish.
Otherwise the matrix was assumed to be irreducible.

For each Lefkovitch matrix \(L\), the dominant eigenvalue \(\lambda\)
and its normalised left and right eigenvectors were first computed,
with \(vu=1\), and \(r=\log\lambda\). The demographic coefficients
\(m_j\), \(b_j\), and \(c_j\), together with the index sets
\(\mathcal F\) and \(\mathcal R\) for positive coefficients, were then
extracted from
\eqref{eq:Lefkovitch}.

The quantities
\[
R_j,\quad q_j,\quad Q_i,\quad
\eta_i,\quad \rho_i,\quad T,\quad H,\quad \Phi
\]
and the transition--retention decompositions of \(H\) and \(\Phi\)
were evaluated directly from the formulas of Sections~2 and~3.
The identity
\[
r=H+\Phi
\]
was used as an internal numerical check.

For changes to existing links, the quantities \(C_j\), \(D_i\),
\(A_j\), and \(B\) were evaluated for the positive coefficients,
and the sensitivities of \(r\), \(H\), and \(\Phi\) were computed from
\eqref{GrowthSensitivityFormulae},
\eqref{SensitivityM}--\eqref{SensitivityC}, and
\[
\frac{\partial\Phi}{\partial\theta}
=
\frac{\partial r}{\partial\theta}
-
\frac{\partial H}{\partial\theta}.
\]
Zero fertility or retention coefficients were not inserted into the
existing-link sensitivity formulas. When a new demographic link was
created
from a zero coefficient, the one-sided asymptotic expressions of
equations~\eqref{BoundaryGrowthFertility}--
\eqref{BoundarySingularSensitivity} were used instead.

\begingroup
\color{black}
For representative new-link calculations, these asymptotic laws were
also checked numerically by one-sided changes, verifying that
\[
\frac{H(\varepsilon)-H(0)}{-\varepsilon\log\varepsilon}\to\kappa,
\qquad
\frac{\Phi(\varepsilon)-\Phi(0)}{\varepsilon\log\varepsilon}\to\kappa,
\]
where \(\kappa\) is the appropriate coefficient
\(\kappa_k^{(m)}\) or \(\kappa_k^{(c)}\) defined above for fertility
or retention link creation, respectively, while
$\bigl(r(\varepsilon)-r(0)\bigr)/\varepsilon$ converges to the
finite coefficient in \eqref{BoundaryGrowthFertility} or
\eqref{BoundaryGrowthRetention}.
\endgroup

For positive demographic coefficients, the reported parameter-scaled
sensitivities were
\[
e_\theta^{(X)}
=
\theta\frac{\partial X}{\partial\theta},
\qquad
X\in\{r,H,\Phi\}.
\]
For \(X=r\), this is also the conventional elasticity of the dominant
eigenvalue. Equation~\eqref{UniformScalingIdentities} was used as an
additional exact check on the computed
parameter-scaled sensitivities. Finally, all analytical sensitivities
for existing links used in the empirical examples were
independently checked against finite differences, with agreement to
the numerical precision reported in the tables and figures.

\endgroup

\bibliographystyle{abbrv}
\bibliography{BibloH2}

\end{document}